\documentclass{article}

\usepackage{amsmath}
\usepackage{bm, icomma}
\usepackage{setspace}
\usepackage{bbm}
\usepackage{amsthm}

\usepackage[shortlabels]{enumitem}

\newtheorem*{assumption*}{\assumptionnumber}
\providecommand{\assumptionnumber}{}
\makeatletter
\newenvironment{assumption}[2]
 {%
  \renewcommand{\assumptionnumber}{Assumption #1#2}%
  \begin{assumption*}%
  \protected@edef\@currentlabel{#1#2}%
 }
 {%
  \end{assumption*}
 }
\makeatother
\AfterEndEnvironment{assumption}{\noindent\ignorespaces}

\usepackage{censor}

\usepackage{PRIMEarxiv}

\usepackage[utf8]{inputenc} 
\usepackage[T1]{fontenc}    
\usepackage{hyperref}       
\usepackage{url}            
\usepackage{booktabs}       
\usepackage{amsfonts}       
\usepackage{nicefrac}       
\usepackage{microtype}      
\usepackage{lipsum}
\usepackage{fancyhdr}       
\usepackage{graphicx}       
\graphicspath{{media/}}     

\newtheorem{theorem}{Theorem}
\usepackage{amsmath,amssymb,amstext} 
\usepackage[linesnumbered,ruled,vlined]{algorithm2e}
\SetKwInput{KwInput}{Input}                
\SetKwInput{KwOutput}{Output}

\title{On Flexible Inverse Probability of Treatment and Intensity Weighting: Informative Censoring, Variable Inclusion, and Weight Trimming}

\author{
  Grace Tompkins \thanks{Corresponding author} \\
  Department of Statistics and Actuarial Sciences\\
  University of Waterloo \\
  Waterloo, ON\\
  \texttt{grace.tompkins@uwaterloo.ca} \\
   \And
  Joel A. Dubin \\
  Department of Statistics and Actuarial Sciences\\
  University of Waterloo \\
  Waterloo, ON\\
  \texttt{jdubin@uwaterloo.ca} \\
  \And
  Michael Wallace \\
  Department of Statistics and Actuarial Sciences\\
  University of Waterloo \\
  Waterloo, ON\\
  \texttt{michael.wallace@uwaterloo.ca} \\
}

\begin{document}
\maketitle

\begin{abstract}
Many observational studies feature irregular longitudinal data, where the observation times are not common across individuals in the study. Further, the observation times may be related to the longitudinal outcome. In this setting, failing to account for the informative observation process may result in biased causal estimates. This can be coupled with other sources of bias, including non-randomized treatment assignments and informative censoring. This paper provides an overview of a flexible weighting method used to adjust for informative observation processes and non-randomized treatment assignments. We investigate the sensitivity of the flexible weighting method to violations of the noninformative censoring assumption, examine variable selection for the observation process weighting model, known as inverse intensity weighting, and look at the impacts of weight trimming for the flexible weighting model. We show that the flexible weighting method is sensitive to violations of the noninformative censoring assumption and show that a previously proposed extension fails under such violations. We also show that variables confounding the observation and outcome processes should always be included in the observation intensity model. Finally, we show that weight trimming should be applied in the flexible weighting model when the treatment assignment process is highly informative and driving the extreme weights. We conclude with an application of the methodology to a real data set to examine the impacts of household water sources on malaria diagnoses. 
\end{abstract}


\section{Background}

Researchers in healthcare often have access to observational data in the form of medical records or clinical data. Such data are often longitudinal, where individuals receiving care have repeated observations recorded over time. Observational data provide unique challenges for the analysis of longitudinal data when individuals do not have a prescribed observation schedule. In this setting, we often encounter \emph{irregular} longitudinal data where the observation times vary between individuals, and in extreme cases are completely unique. Methods like \emph{generalized linear mixed-effects models} (GLMMs) and, in some cases, \emph{generalized estimating equations} (GEEs) can handle irregular longitudinal data. However, these methods may produce biased estimates of the outcome model parameters if the observation times are \emph{informative}, meaning the observation times are related to the longitudinal outcome. For example, we may have access to observational clinical data where patients with more severe symptoms may be more likely to visit a clinic and be seen by a physician, who then records their data into a database. If we were to analyze these clinical data without accounting for the observation process, we may obtain biased estimates of any treatment effects due to the informative observation process. Methods for handling specific cases of informative observation processes have previously been proposed in the literature, and include \emph{inverse intensity weighting} (IIW) \cite{Lin2004, Buzkova2007} and other semi-parametric methods involving correlated random effects \cite{Liang2009, sun2007regression, sun2011semiparametric, sun2011regression, cai2012time, song2012regression, sun2012joint}.  

If a treatment effect is to be estimated, we must also consider the nature of the treatment assignment process for observational longitudinal data.  In most observational data sets, the treatment or intervention of interest will not be assigned randomly. For example, patients with more severe symptoms may be more likely to receive treatment, or, say a higher dose of a given treatment.  In this setting, the treatment groups may have systematic differences that do not make them directly comparable, and these differences may bias estimates of causal effects if not accounted for. When we cannot rely on randomization to balance the treatment groups, \emph{inverse probability of treatment weighting} (IPTW) \cite{RR1983} can be employed to reduce the bias of causal estimates, such as the \emph{average treatment effect} (ATE). 

In practice, we may see data with informative observation processes where the treatment assignment is not randomized. For example, if we consider the analysis of data from a clinical database, sicker patients may be more likely to visit the clinic and have an observation recorded, and also may be more likely to be prescribed a treatment. In this setting, we have two potential sources of bias; one from the informative observation process, and one from the non-randomized treatment assignment process. To simultaneously account for both sources of bias, we can employ a weighting method proposed by Coulombe et al. \cite{coulombe2021}, which we refer to as the \emph{flexible inverse probability of treatment and intensity weighting} (FIPTIW). This method combines IIW and IPTW in an intuitive way that is simple to implement using standard statistical software. 

The aim of this paper is to provide an overview and practical guidance for fitting models using FIPTIW by providing novel investigations on the sensitivity of the FIPTIW method. We begin with a discussion of the assumptions on the observation and treatment assignment processes in Section \ref{sec:FIPTIWassumptions}. We then review the existing methods for handling informative observation and treatment assignment processes in Section \ref{sec:FIPTIWmethods}. We follow with three simulation studies in Section \ref{sec:FIPTIWsimulations}. The first investigates the impact of violations of the noninformative censoring assumption on the FIPTIW method, and also investigate if \emph{inverse probability of censoring weighting} (IPCW) \cite{robins2000marginal} can be further included to account for the bias introduced by such violations. Next, we investigate variable inclusion for IIW models to provide practical guidance on which variables to include in IIW (and hence FIPTIW) models. Finally, we investigate the impacts of extreme weights and  weight trimming for FIPTIW. In Section \ref{sec:FIPTIWdatanalysis}, we implement a data analysis of a malaria data set. We conclude with a discussion in Section \ref{sec:ch1discussion}.

\section{Assumptions}\label{sec:FIPTIWassumptions}

Consider a study from $t=0$ to the study end time $\tau$ where $t_{i1}, t_{i2},... , t_{iK_i}$ are the times at which individual $i$ is observed for $i = 1, 2, ..., n$ and  $0 \le t_{i1} < t_{i2} < ... < t_{iK_i} \le \tau$  where $\tau$ is the study end time. That is, each individual in the study has a potentially unique set of $K_i$ observation times.

For the purposes of this paper, we will focus on the following semiparametric marginal outcome model
\begin{equation}\label{eq:outcomemodelch1}
    g(\mu_i(t)) = \boldsymbol{\beta}^T\boldsymbol{X}_i(t),
\end{equation}
with primary interest in the estimation of the parameter vector $\boldsymbol{\beta}$. In Equation (\ref{eq:outcomemodelch1}), $\boldsymbol{X}_i(t) = (X_{i1}(t), X_{i2}(t), \allowbreak ...., X_{ip}(t))^T$ is a vector of $p$ observed covariates for individual $i$ at time $t$, and $\mu_i(t) = E(Y_i(t) | \boldsymbol{X}_i(t))$ where $Y_i(t)$ is the longitudinal outcome of individual $i$ at time $t$, for $i = 1, 2, ..., n$ and $t = t_{i1}, t_{i2}, ..., t_{iK_i}$. Further, $g(\cdot)$ is a known, monotonic, and differentiable link function. We assume the set of covariates $\boldsymbol{X}(t)$ contains a possibly time-varying, binary treatment $D(t)$ for which we'd like to estimate the treatment effect. 

\subsection{Assumptions on the Observation Process}

We explicitly define the observation process to be the underlying process that dictates when individuals have data recorded. In other articles, this has been referred to as the ``visit" \cite{Lin2004, Pullenayegum2012, Pullenayegum2016, Pullenayegum2021, aghababaei2022variable, coulombe2021continuous}, ``participation"\cite{Schmidt2017}, ``assessment" \cite{Smith2022, Pullenayegum2023}, or ``monitoring" \cite{coulombe2021} process. We opt to use ``observation" over other terms as it is more general and also connects irregular longitudinal data analysis to the missing data literature, as the analysis of longitudinal data can often be viewed as a missing data problem \cite{Buzkova2007}.

Assume the primary objective is to model the conditional mean of a longitudinal outcome $Y(t)$ on a set of covariates $\boldsymbol{X}(t)$. We denote the counting process for the number of observations individual $i$ has by time $t$ as $N_i(t) = \sum_{k = 1}^{K_i}\mathbbm{1}_{(t_{ik}\le t)}$, where $\mathbbm{1}_{(E)}$ is the indicator function for event $E$. We let $C_i$ denote the censoring time at which follow-up ceases for individual $i$, such that $C_i \le \tau$. Although individuals may be censored prior to end of the study, we consider the counting process for the counterfactual observation times, denoted $N_i^*(t)$. We relate the censored and uncensored counting processes as $N_i(t) = N_i^*(t \wedge C_i)$, where we define $a \wedge b = min(a,b)$. Let $\boldsymbol{V}_i(t)$ be a set of auxiliary covariates related to the observation process but omitted from the outcome model.

For any arbitrary process $A(t)$, we define $\bar{A}(t) = \{A(s): 0 \le s \le t\}$ as the entire (and potentially counterfactual)  history of the process up to and including time $t$. We denote $\bar{A}(t)^{obs}$ to include only the observed history of $A$ up to and including time $t$. We let $\bar{A}(\infty) = \{A(s): a > 0\}$ be the entire (counterfactual) process that includes times beyond the study end time. Further, we let $dN(t) = N(t) - N(t^{-})$ where $N(t^{-}) = \lim_{s \rightarrow t}N(s)$. That is, $dN_i(t) = 1$ if individual $i$ is observed at time $t$, and is 0 otherwise. $dN_i^*(t)$ is similarly defined, where $dN_i^*(t) = 1$ if individual $i$ is observed at time $t$ in the counterfactual observation times, and is 0 otherwise.  

Throughout this paper, we will refer to five assumptions on the observation process. The first assumption states that the intensity or probability of the observation times can depend on a set of covariates $\boldsymbol{Z}(t)$, which can include the observed history of the outcome covariates $\bar{\boldsymbol{X}}(t)$, the observed history of the auxiliary covariates $\bar{\boldsymbol{V}}_i(t)$, previous observation times $\bar{\boldsymbol{N}}_i(t^-)$, and previous observed outcomes $\bar{\boldsymbol{Y}}^{obs}_i(t^-)$. That is, we assume
\begin{assumption}{O}{1}\label{O1}
Independent sampling:
\begin{equation*}
       E(dN_i^*(t) | \boldsymbol{Z}_i(t), \boldsymbol{X}_i(t), Y_i(t), C_i \ge t) = E(dN_i^*(t) | \boldsymbol{Z}_i(t)).
   \end{equation*}
\end{assumption} 
\noindent As $\boldsymbol{Z}_i(t)$ can contain the outcome model covariates ($\boldsymbol{X}_i(t)$), the observed history of the outcome model covariates ($\bar{\boldsymbol{X}}_i^{obs}(t)$), the observed history of auxiliary covariates related to the probability of being observed but omitted from the outcome model ($\bar{\boldsymbol{V}}_i^{obs}(t)$), information on past observation times ($\bar{\boldsymbol{N}}_i(t^-)$), and past observed outcomes prior to time $t$ ($\bar{\boldsymbol{Y}}^{obs}_i(t^-)$), we can re-write Assumption \ref{O1} as
   \begin{equation*}
   \begin{aligned}
       E(dN_i^*(t) |  \bar{\boldsymbol{X}}_i^{obs}(t), \bar{\boldsymbol{V}}_i^{obs}(t), \bar{\boldsymbol{N}}_i(t^-), \bar{\boldsymbol{Y}}^{obs}_i(t^-), &\boldsymbol{X}_i(t), Y_i(t), C_i \ge t)\\ = E(dN_i^*(t) | &\boldsymbol{X}_i(t),\bar{\boldsymbol{X}}_i^{obs}(t), \bar{\boldsymbol{V}}_i^{obs}(t), \bar{\boldsymbol{N}}_i(t^-), \bar{\boldsymbol{Y}}^{obs}_i(t^-)).       
   \end{aligned}
   \end{equation*}
This implies $N_i^*(t) \perp  Y_i(t), C_i \ge t | \bar{\boldsymbol{X}}_i^{obs}(t), \bar{\boldsymbol{V}}_i^{obs}(t), \bar{\boldsymbol{N}}_i(t^-), \text{ and } \bar{\boldsymbol{Y}}^{obs}_i(t^-)$. 

Next, we make an assumption on the censoring times such that 
\begin{assumption}{O}{2}\label{O2}
Noninformative censoring:
\begin{equation*}
    E(Y_i(t) | \boldsymbol{X}_i(t), C_i \ge t) = E(Y_i(t) | \boldsymbol{X}_i(t)).
\end{equation*}
\end{assumption}
\noindent Assumption \ref{O2} may not be met in some applications. For example, when analyzing the relationship between a treatment and a health outcome in an observational study, sicker (or in some cases, healthier) patients may drop out prior to the end of the study. In this setting, censoring times are related to longitudinal outcome, which is related to the observation times. We will investigate the implications violations of this assumption in Section \ref{sec:censoringassumptionsim}.

We also assume 
\begin{assumption}{O}{3}\label{O3}
Separability: The outcome and observation times model parameters are separable (i.e., the models do not share parameters).
\end{assumption}

We also assume
\begin{assumption}{O}{4}\label{O4}
Correct specification of the observation intensity model: The model for the observation intensity is correctly specified.
\end{assumption}

And finally, we assume
\begin{assumption}{O}{5}\label{O5}
Completely observed observation-level covariates: The covariates related to the observation process are known at all possible observation times.
\end{assumption}
This assumption means that at any possible observation time (including those not observed), we can obtain the value of $\boldsymbol{Z}_i(t)$ for all individuals. When $\boldsymbol{Z}_i(t)$ does not contain any time-varying covariates, this assumption is automatically met. However, the use of time-varying covariates in $\boldsymbol{Z}_i(t)$ can complicate analysis when these covariates are not observed in continuous time. When $\boldsymbol{Z}_i(t)$ is not observed completely, Buzkova and Lumley \cite{Buzkova2007} recommend carrying the last observed value forward for unobserved $\boldsymbol{Z}_i(t)$. However, one must be cautious when employing this method when a non-trivial proportion of $\boldsymbol{Z}_i(t)$ are missing, as it can bias the parameters estimated in \emph{proportional hazards} (PH) models \cite{molenberghs2002prediction, molnar2008does, andersen2003attenuation, cao2021proportional} and longitudinal models \cite{lachin2016fallacies, saha2009bias, lane2008handling}.

When assumptions \ref{O1} to \ref{O5} are met, we will refer to the observation process as \emph{conditionally ignorable}. That is, conditional on the observed history of the covariates, the observation process can adequately be accounted for in the analysis to obtain unbiased and consistent estimates of the outcome model parameters.

\subsection{Assumptions on the Treatment Assignment Process}

We explicitly define the treatment assignment process as the underlying mechanism that determines whether individuals are treated or not over time. As previously discussed, most observational studies do not have treatments randomly assigned to groups of individuals. That is, confounding can be present in the treatment assignment if it was not randomized. In this setting, the treatment and control groups may systematically differ from each other and estimates of the ATE and other causal quantities may be biased. We assume that the probability of being assigned to the treatment group at any time is related to a set of possibly time-varying covariates $\boldsymbol{W}(t)$. 

 To discuss the assumptions we will make on the treatment assignment process, we first introduce the \emph{potential outcomes framework} as developed by Rosenbaum and Rubin \cite{RR1983}. Recall $Y_i(t)$ is the observed outcome for individual $i$ at time $t$. We also consider the potential outcomes that individual $i$ would have experienced had they been assigned to a specific treatment. We denote $Y_i^{(1)}(t)$ and $Y_i^{(0)}(t)$ to be the potential outcomes under treatment and control, respectively.  We note that these quantities are counterfactual as we cannot observe both potential outcomes simultaneously in practice. 

Regardless of the study design, the estimation of causal effects is a comparison of the potential outcomes. We often are interested in estimating the ATE, which we define as $E\{Y_i^{(1)}(t) - Y_i^{(0)}(t)\}$. To use IPTW, we require a number of assumptions in this framework. The first assumption is
\begin{assumption}{T}{1}\label{T1}
Strongly ignorable treatment assignment (SITA): the treatment assignment is conditionally independent of the potential outcomes given the observed covariates $\boldsymbol{W}(t)$, such that
\begin{equation*}
    Y_i^{(0)}(t), Y_i^{(1)}(t) \perp D_i(t) | \boldsymbol{W}_i(t).
\end{equation*}
\end{assumption}
\noindent Implicit in Assumption \ref{T1} is the assumption that there are no unobserved confounders.

Next, we assume
\begin{assumption}{T}{2}\label{T2}
Stable unit treatment value assumption (SUTVA): the treatment assignment of a given individual does not affect the outcome of another individual.
\end{assumption}
\noindent The SUTVA assumption may not hold if, for example, the vaccination status of an individual impacts the probability of transmitting the disease to someone they are in contact with. 

We also assume
\begin{assumption}{T}{3}\label{T3}
Consistency: the observed outcome under a specific treatment is equal to that treatment's potential outcome, and we can only observe one potential outcome at a time. 
\end{assumption}
\noindent Under Assumption \ref{T3}, under a binary treatment we can write the observed outcome as a function of the potential outcomes as
\begin{equation*}
    Y_i(t) = D_iY_i^{(1)}(t) + (1 - D_i)Y_i^{(0)}(t).
\end{equation*}

We further assume
\begin{assumption}{T}{4}\label{T4}
Positivity: the probability of receiving a given treatment is greater than zero for every individual.
\end{assumption}

We  finally assume
\begin{assumption}{T}{5}\label{T5}
Correct specification of the propensity score model: the model for the propensity score is correctly specified.
\end{assumption}

When Assumptions \ref{T1} to \ref{T5} hold, we refer to the treatment assignment process as \emph{conditionally ignorable}. That is, conditional on the observed history of the covariates, the treatment assignment process can be adequately accounted for in the analysis through IPTW to provide consistent and unbiased estimates of causal quantities like the ATE.

\section{Methods}\label{sec:FIPTIWmethods}

\subsection{Inverse Intensity Weighting}\label{sec:iiw}

Under Assumptions \ref{O1} to \ref{O5}, we consider the observation process to be conditionally ignorable and we can employ the inverse intensity weighted generalized estimating equation (IIW-GEE) to obtain estimates of the parameters of interest in the outcome model in Equation (\ref{eq:outcomemodelch1}) (ignoring the treatment assignment process for now). Estimating an IIW-GEE involves a two-step process where we first estimate the inverse intensity weights in an observation times model and then use the estimated weights to obtain weighted estimates of the quantities of interest from the outcome model.

To use IIW, we model the uncensored observation times as
\begin{equation}\label{eq:obstimemodelch1}
    E\left\{ dN_i^*(t)|\boldsymbol{Z}_i(t)\right\}=  \lambda_0(t)\exp \left\{\boldsymbol{\gamma}^T\boldsymbol{Z}_i(t) \right\},
\end{equation}
where $\lambda_0(\cdot)$ is an unspecified non-decreasing function, $\boldsymbol{Z}_i(t)$ is a vector of covariates for the observation times model, and $\boldsymbol{\gamma}$ is the corresponding parameter vector. Recall $\boldsymbol{Z}_i(t)$ may include outcome model covariates ($\boldsymbol{X}(t)$), past observed outcomes ($\bar{\boldsymbol{Y}}^{obs}(t^-)$), information about previous observation times ($\bar{\boldsymbol{N}}_i(t^-)$), and auxiliary covariates that are not included in the outcome model but related to the observation times ($\boldsymbol{V}_i(t)$). We emphasize that  $\boldsymbol{Z}_i(t)$ must be known at all times (Assumption \ref{O5}).

We define the inverse intensity weights, for the $i$th individual at time $t$, as
\begin{equation}\label{eq:weighteq}
    w_i^{IIW}(t; \boldsymbol{\gamma}, h) = \frac{h(\boldsymbol{X}_i(t))}{\exp \left\{ \boldsymbol{\gamma}^T\boldsymbol{Z}_i(t) \right\}},
\end{equation}
where $h(\cdot)$ can be any positive function of the outcome model covariates $\boldsymbol{X}_i(t)$. In the observation times model we can include (and must specify) various functional forms of the covariates $\boldsymbol{Z}_i(t)$ including interactions and higher-order terms \cite{Lin2004}. These inverse intensity weights are proportional to the probability of individual $i$ having an observation at time $t$ relative to the other individuals, under the model in Equation (\ref{eq:obstimemodelch1}) \cite{Buzkova2007}.

In practice, these weights must be estimated from the data as the parameter vector $\boldsymbol{\gamma}$ is unknown in the observation times model. To estimate the weights, we use semiparametric models \cite{Lin2004}. Buzkova and Lumley \cite{Buzkova2007} showed that the following estimation function can be used to estimate the parameter vector $\boldsymbol{\gamma}$:
\begin{equation*}
    U^{\dagger}(\boldsymbol{\gamma}) = \sum_{i = 1}^n\int_0^{\tau} \{\boldsymbol{Z}_i(t) - \bar{\boldsymbol{Z}}(t, \boldsymbol{\gamma})\}dN_i(t),
\end{equation*}
where $\bar{\boldsymbol{Z}}(t, \gamma)$ is a weighted average of $\boldsymbol{Z}$ at time $t$, such that
\begin{equation*}
    \bar{\boldsymbol{Z}}(t, \boldsymbol{\gamma}) = \sum_{i = 1}^n \boldsymbol{Z}_i(t) \frac{\exp\left\{\boldsymbol{\gamma}^T\boldsymbol{Z}_i(t)\right\}I(C_i \ge t)}{\sum_{j = 1}^n \exp\left\{\boldsymbol{\gamma}^T\boldsymbol{Z}_j(t)\right\}I(C_j \ge t)}.
\end{equation*}
The solution to $U^{\dagger}(\boldsymbol{\gamma}) = 0$, $\hat{\boldsymbol{\gamma}}$, is a consistent estimator of the parameter vector $\boldsymbol{\gamma}$ in the observation times model \cite{Buzkova2007}. Because of its form, the Cox PH model can be used to estimate the intensity as in Equation (\ref{eq:obstimemodelch1}), however we note that the estimator for the asymptotic variance of $\hat{\boldsymbol{\gamma}}$ will differ from the asymptotic variance of parameters estimated using a Cox PH model \cite{Buzkova2007}.

Recall that in the model for the weights in Equation (\ref{eq:weighteq}), we also require the specification of a function $h(\cdot)$. A convenient choice of the numerator of the  weights is $h(\boldsymbol{X}_i(t)) = 1$. We refer to such weights as \emph{non-stabilized} weights. We can also use \emph{stabilized} weights, as proposed by Buzkova and Lumley \cite{Buzkova2007}, where $h(\boldsymbol{X}_i(t)) = \exp\{ \boldsymbol{\delta}^T\boldsymbol{X}_i(t)\}$ and $\boldsymbol{\delta}$ is the parameter vector in the observation times model using only the outcome model covariates $\boldsymbol{X}_i(t)$. The stabilized weights are typically preferred as they can achieve a smaller variance than the non-stabilized weights \cite{Buzkova2007}. Further, if the observation process is uninformative, the IIW-GEE simplifies to an independent unweighted GEE using the stabilized weights.  

To estimate the outcome model in Equation (\ref{eq:outcomemodelch1}) when the observation process is conditionally ignorable (while ignoring any treatment assignment processes) we can use a weighted GEE, which involves specifying a number of components. First, we need to specify an appropriate link function $g(\cdot)$ for the semiparametric marginal model in Equation (\ref{eq:outcomemodelch1}). Second, we require the specification of the conditional variance of each observation, given the covariates. We assume this quantity is dependent on the mean through
\begin{equation}\label{eq:variancefun}
    Var\left\{Y_i(t) | \boldsymbol{X}_i(t)\right\} = \phi v(\mu_i(t)),
\end{equation}
where $v(\cdot)$ is a known variance function and $\phi$ is a positive scale parameter, which may need to be estimated. 

The estimating function for the outcome model, as motivated by the GEE, is
\begin{equation*}
\begin{aligned}
    U(\boldsymbol{\beta}; \hat{\boldsymbol{\gamma}}, h) &= \sum_{i = 1}^n \int_{0}^{\tau} \boldsymbol{X}_i(t) \left\{ \frac{dg(\mu)}{d\mu}\bigg\rvert_{\mu_i(t; \boldsymbol{\beta})} \right\}^{-1} v(\mu_i(t; \boldsymbol{\beta}))^{-1}\\
    &\times \{Y_i(t) - \mu_i(t; \boldsymbol{\beta}) \}\frac{h(\boldsymbol{X}_i(t))}{\lambda_i(t; \boldsymbol{\widehat{\gamma}}, h)}dN_i(t),
\end{aligned}
\end{equation*}
which resembles a GEE with an independent working correlation structure and subject-specific inverse intensity weights. Correct specification of the variance function $v(\mu_i(t; \boldsymbol{\beta}))$ increases the efficiency of the estimator for $\boldsymbol{\beta}$, however is not required to obtain consistent estimates, asymptotic normality, or the validity of its covariance estimator \cite{Buzkova2007}.

\subsection{Flexible Inverse Probability of Treatment and Intensity Weighting}

As previously discussed, we may also have to adjust for the treatment assignment process, along with the observation process, when treatments were not randomly assigned. When the observation process is conditionally ignorable (assumptions \ref{O1} to \ref{O5} hold) and the treatment assignment process is conditionally ignorable (assumptions \ref{T1} to \ref{T5} hold), we can employ the flexible weighting method proposed by Coulombe et al. \cite{coulombe2021}, which we refer to as FIPTIW. FIPTIW combines IIW and IPTW to create a pseudopopulation in which both the observation  and treatment assignment processes can be ignored. 

We first define the IPTW weight for individual $i$ at time $t$ as
\begin{equation*}
    w_i^{IPTW}(t; \boldsymbol{\alpha}, \pi) = \frac{1}{\mathbbm{1}_{(D_i(t) = 1)}\pi(\boldsymbol{W}_i(t); \boldsymbol{\alpha}) + \mathbbm{1}_{(D_i(t) = 0)}(1 - \pi(\boldsymbol{W}_i(t); \boldsymbol{\alpha}))},
\end{equation*}
where $\pi(\boldsymbol{W}_i(t); \boldsymbol{\alpha}) = \text{Pr}(D_i(t) = 1 | \boldsymbol{W}_i(t))$ is the probability of being in the treatment group at time $t$ (conditional on covariates), which is also known as the \emph{propensity score}. The propensity scores (and thus the IPTW weights) are unknown in practice and must be estimated from the available data. While many choices of model exist for estimating the probability of receiving a binary treatment, common methods include logistic regression or tree-based methods such as \emph{generalized boosted models} (GBMs) \cite{Austin2011}. We can model the probability of treatment assignment and obtain estimates of $\widehat{\boldsymbol{\alpha}}$ and then estimate the propensity scores as $\pi(\boldsymbol{W}_i(t); \widehat{\boldsymbol{\alpha}})$. Plugging in the estimated propensity scores, we estimate the IPTW weights as
\begin{equation*}
    \widehat{w}_i^{IPTW}(t; \widehat{\boldsymbol{\alpha}}, \pi) = \frac{1}{\mathbbm{1}_{(D_i(t) = 1)}\pi(\boldsymbol{W}_i(t); \widehat{\boldsymbol{\alpha}}) + \mathbbm{1}_{(D_i(t) = 0)}(1 - \pi(\boldsymbol{W}_i(t); \widehat{\boldsymbol{\alpha}}))},
\end{equation*}
where again, $\mathbbm{1}(\cdot)$ is the indicator function. 

The FIPTIW weights are calculated by multiplying the IIW and IPTW weights together. That is, the FIPTIW weight for individual $i$ at time $t$ is
\begin{equation}
\begin{aligned}
    w_i^{FIPTIW}(t; \boldsymbol{\alpha}, \boldsymbol{\gamma}, \pi, h) &= w_i^{IPTW}(t; \boldsymbol{\alpha}, \pi) \times w_i^{IIW}(t; \boldsymbol{\gamma}, h).\\
\end{aligned}
\end{equation}
However, as each of the individual weights in the FIPTIW are unknown in practice, the FIPTIW weight must also be estimated from the data. We estimate it by
\begin{equation*}
    \widehat{w}_i^{FIPTIW}(t; \widehat{\boldsymbol{\alpha}}, \widehat{\boldsymbol{\gamma}}, \pi, h) =   \widehat{w}_i^{IIW}(t;  \widehat{\boldsymbol{\gamma}},  h) \times \widehat{w}_i^{IPTW}(t;  \widehat{\boldsymbol{\alpha}}, \pi),
\end{equation*}
where the IIW weights, $\widehat{w}_i^{IIW}(t; \widehat{\boldsymbol{\gamma}}, h)$,  can be estimated as in Section \ref{sec:iiw}.

For simplicity, the remainder of this paper focuses on the setting where the parametric form of the intercept function is specified, as in the work on IIW by Buzkova and Lumley \cite{Buzkova2007}, however we note that the intercept could be estimated using splines as in Coulombe et al. \cite{coulombe2021}. We explicitly define a weighted GEE (where the working correlation structure is set to independent) which incorporates the FIPTIW weights into the model. We refer to this GEE as the FIPTIW-GEE, with the estimating equation specified as
\begin{equation*}
\begin{aligned}\label{eq:estimatingknown}
    U(\boldsymbol{\beta}; \boldsymbol{\alpha}, \boldsymbol{\gamma}, \pi, h) &= \sum_{i = 1}^n \int_{0}^{\tau} \boldsymbol{X}_i(t) \left\{ \frac{dg(\mu)}{d\mu}\bigg\rvert_{\mu_i(t; \boldsymbol{\beta})} \right\}^{-1} v(\mu_i(t; \boldsymbol{\beta}))^{-1}\\
    &\times \{Y_i(t) - \mu_i(t; \boldsymbol{\beta}) \}w_i^{FIPTIW}(t; \boldsymbol{\alpha}, \boldsymbol{\gamma}, \pi, h)dN_i(t).
\end{aligned}
\end{equation*}
As the FIPTIW weights are unknown, we can then use the estimated weights in the estimating equation as
\begin{equation}\label{eq:FIPTIWgee}
\begin{aligned}
    U(\boldsymbol{\beta}; \widehat{\boldsymbol{\alpha}}, \widehat{\boldsymbol{\gamma}}, \pi, h) &= \sum_{i = 1}^n \int_{0}^{\tau} \boldsymbol{X}_i(t) \left\{ \frac{dg(\mu)}{d\mu}\bigg\rvert_{\mu_i(t; \boldsymbol{\beta})} \right\}^{-1} v(\mu_i(t; \boldsymbol{\beta}))^{-1}\\
    &\times \{Y_i(t) - \mu_i(t; \boldsymbol{\beta}) \}\widehat{w}_i^{FIPTIW}(t; \widehat{\boldsymbol{\alpha}}, \widehat{\boldsymbol{\gamma}}, \pi, h)dN_i(t).
\end{aligned}
\end{equation}

Coulombe et al. \cite{coulombe2021} did not make any assumptions about the dependence between the observation and treatment assignment processes. However, when the treatment assignment and observation processes are conditionally independent given $\boldsymbol{X}_i(t), \boldsymbol{Z}_i(t), \boldsymbol{W}_i(t), Y_i(t),$ and the censoring time $C_i$, the weights intuitively reflect the inverse probability of individual $i$ having an observation at time $t$ and receiving the treatment they were assigned, relative to the other individuals. Further, under this assumption, we can show that the estimating equation in Equation (\ref{eq:FIPTIWgee}) is unbiased by the following theorem:
\begin{theorem}\label{thm:unbiased}
The FIPTIW estimating equation in Equation (\ref{eq:FIPTIWgee}) has zero mean at the point of true parameters $\{\boldsymbol{\beta}, \boldsymbol{\alpha}, \boldsymbol{\gamma}, \pi, h\}$ for any function $h(\cdot)$ of covariates $\boldsymbol{X}_i(t)$.
\end{theorem}

\begin{proof}
See Supplementary Material \ref{sec:proofs}.
\end{proof}

We denote $\widehat{\boldsymbol{\beta}}$ to be the estimator of parameter vector $\boldsymbol{\beta}$ in Equation (\ref{eq:outcomemodelch1}) where $\widehat{\boldsymbol{\beta}}$  is a solution to $U(\boldsymbol{\beta}; \widehat{\boldsymbol{\alpha}}, \widehat{\boldsymbol{\gamma}}, \pi, h) = 0$. As we assume that the observations are independent and identically distributed and have shown that the estimating equation is unbiased in Theorem \ref{thm:unbiased}, then it follows that the solution $\widehat{\boldsymbol{\beta}}$ from $U(\widehat{\beta}; \widehat{\boldsymbol{\alpha}}, \widehat{\boldsymbol{\gamma}}, \pi, h) = 0$ is a consistent estimator of $\boldsymbol{\beta}$, as the estimators for the parameters in  $w_i^{FIPTIW}(t; \widehat{\boldsymbol{\alpha}}, \widehat{\boldsymbol{\gamma}}, \pi, h)$  are each consistent\cite{Tsiatis2006}. That is, the estimator $\widehat{\boldsymbol{\beta}}$ satisfies
\begin{equation*}
    \begin{aligned}
    \frac{1}{n}\sum_{i = 1}^n &\int_{0}^{\tau} \boldsymbol{X}_i(t) \left\{ \frac{dg(\mu)}{d\mu}\bigg\rvert_{\mu_i(t; \widehat{\boldsymbol{\beta}})} \right\}^{-1} v(\mu_i(t; \widehat{\boldsymbol{\beta}}))^{-1}\\
    &\times \{Y_i(t) - \mu_i(t; \widehat{\boldsymbol{\beta}}) \}\widehat{w}_i^{FIPTIW}(t; \widehat{\boldsymbol{\alpha}}, \widehat{\boldsymbol{\gamma}}, \pi, h)dN_i(t) = 0.
\end{aligned}
\end{equation*}
As this is a weighted average with mean zero, by the Law of Large Numbers it follows that $\frac{1}{n}U(\widehat{\boldsymbol{\beta}}, \widehat{\boldsymbol{\alpha}}, \widehat{\boldsymbol{\gamma}}, \pi, h) \overset{p}{\to} U(\boldsymbol{\beta}, \boldsymbol{\alpha}, \boldsymbol{\gamma}, \pi, h)$ and thus the estimator is consistent. The central limit theorem (CLT) can also be employed to prove the asymptotic normality of $\widehat{\boldsymbol{\beta}}$. Coulombe et al. \cite{coulombe2021} derived the asymptotic variance of the FIPTIW estimator. 

Coulombe et al. \cite{coulombe2021} also performed sensitivity analyses to examine the impact of having correlated treatment assignment confounders, having the same confounders in both the observation model and treatment assignment model, and model misspecification. In these simulations, the performance of the FIPTIW-GEE was not largely impacted by correlated treatment assignment confounders or having the same set of confounders affecting both the observation and treatment assignment processes. The FIPTIW method was also shown to be relatively insensitive to misspecification of the outcome model. However, the FIPTIW method was shown to be extremely sensitive to model misspecification when the observation process depended on non-linear functions of the covariates. 

There are still various questions surrounding the FIPTIW method which are yet to be answered in the literature. First, it is known that censoring/dropout can bias model parameters when it is related to the longitudinal outcome\cite{Ma2005}, however the impacts of violating the noninformative censoring assumption have yet to be explored for the FIPTIW method. Coulombe et al. \cite{coulombe2021} noted that it would be possible to incorporate IPCW weights to account for violations of the noninformative censoring assumption. IPCW weights can be estimated by first fitting a Cox PH model to estimate the censoring hazard. Including IPCW weights into the model may remove some of the bias introduced by the censoring mechanism, however it will not account for the poor estimation of the IIW weights due to the violation of this assumption as the noninformative censoring assumption is necessary for sufficient estimation of the IIW weights. To the best of our knowledge, there have been no papers investigating the inclusion of IPCW weights into the FIPTIW model. 

Second, it has been shown in the causal inference literature that propensity score models for IPTW weights should include true treatment confounders (covariates related to both the treatment assignment and outcome) and covariates related only to the outcome \cite{Brookhart2006}. Including both of these types of covariates in propensity score models has been shown to minimize the mean squared error of the outcome model parameter estimates \cite{Brookhart2006}. Further, including covariates that are only predictive of treatment assignment may inflate the variance of the estimator \cite{Brookhart2006}. However, variable inclusion for IIW (and hence FIPTIW) models has not yet been investigated in the literature. 

Finally, extreme IPTW weights can occur when the estimated propensity scores are close to 0 or 1 (near violations of Assumption \ref{T4}). It has been shown that these extreme weights can lead to an increase in variance of the estimated ATE \cite{Stuart2010}. One solution to handling extreme weights when using IPTW is weight trimming (also called truncation) where weights above a certain threshold are set to a maximum value \cite{potter1993, scharfstein1999}. The threshold is often determined using percentile cut points  where weights above the $p_a$th percentile are set to the value of the percentile $p_a$ \cite{cole2008constructing}. Similarly, we can set weights below the ($1-p_b$)th percentile to the value of the ($1-p_b$)th percentile. Weight trimming has been shown to improve the estimation of IPTW weights using logistic regression, but not when using classification and regression trees (CART) or random forests \cite{lee2011weight}. Although the FIPTIW method can produce extreme weights from either process, the impacts of extreme weights and weight trimming have yet to be examined for FIPTIW. 

In the following section, we aim to fill in the gaps in the existing literature by providing preliminary investigations on the impacts of violations of the noninformative censoring assumption and the inclusion of IPCW weights into the FIPTIW model in Section \ref{sec:censoringassumptionsim},  variable inclusion in intensity models in Section \ref{sec:varinclusion}, and weight trimming in Section \ref{sec:weighttrimming}.

\section{Simulation Studies}\label{sec:FIPTIWsimulations}

\subsection{Data Generating Mechanisms}\label{sec:datagen}

The data generating mechanisms presented in this section are based on the simulation studies presented in Buzkova and Lumley \cite{Buzkova2007}. We consider the scenario where we wish to determine the ATE of a time-invariant treatment $D$ on a continuous longitudinal outcome $Y(t)$. We note that in many scenarios, the estimation of a time-varying treatment may be of interest, however we limit our simulations to the setting where the treatment is time-invariant. Coulombe et al. \cite{coulombe2021} performed various studies using time-varying treatment and confounders, with similar results to the time-invariant setting. We let $D_i = 1$ indicate that individual $i$ is in the treatment group and $D_i = 0$ indicate that individual $i$ is in the control for the duration of the study. We also consider three other covariates,  $W$, $G(t)$, and $Z$, which are each specified below, that may be related to the observation and/or treatment assignment processes, where $0 \le t \le \tau$. 

For each of the $n$ individuals, we simulate covariates such that $G(t) = W\log(t)$ where $W \sim Unif(0,1)$, $D$ is a binary covariate with probability $\pi$, and $Z \sim N(\mu_{z,0},\sigma_{z,0}^2)$ if $D = 0$ or $Z \sim N(\mu_{z, 1},\sigma_{z, 1}^2)$ if $D= 1$. We let $\mu_{z,0} = 2$, $\sigma_{z,0}^2 = 1$, $\mu_{z,1} = 0$, and $\sigma_{z,1}^2 = 0.5$. The probability of being in the treated group $\pi$ will vary in each simulation study. We then generate the outcome as
\begin{equation}\label{eq:simoutcomenormal}
    Y_i(t) = \mu(t) + \beta_1D_i  + \beta_2(G_i(t) - E\{G_i(t)|D_i\}) + \beta_3(Z_i - E\{Z_i|D_i\}) + \phi_i + \epsilon_i(t),
\end{equation}
where $\mu(t) = (2 - t)$, $\epsilon_i(t) \sim N(0, 1)$ is the random error term (for each individual and time) and $\phi_i \sim N(0, 0.25)$ is an individual random effect. We let $\beta_1 = 0.5$, $\beta_2 = 2$, and $\beta_3 = 1$. We simulate a random effects model to create a more realistic scenario for longitudinal data, however the interest lies in the estimation of the marginal model. As the outcome $Y(t)$ is continuous, modelling via a marginal model should have very little effect on the bias of the estimated model parameters $\boldsymbol{\beta}$, even though the data was generated using a subject-specific (mixed effects) model. By centering the auxiliary covariates, we allow these covariates to be related to the longitudinal outcome but omitted from the marginal model of interest
\begin{equation}\label{eq:simulationmarginalmodel}
    E\{Y_i(t)|D_i\} = (2-t) + \beta_1D_i .
\end{equation}
The estimation of $\beta_1$, the average treatment effect, is of primary interest for each simulation study. We assume the form of the intercept $\mu(t)$ is known and treat it as an offset term in the model. Although the simulations assume the outcome is normally distributed, we note that we can also simulate data with non-identity link functions, as in Buzkova and Lumley \cite{Buzkova2007}. 

We also consider a model for the observation times, which has an intensity specified as
\begin{equation}\label{eq:simintensity}
    \lambda_{i,obs}(t) = \nu_i \lambda_0(t)\exp\{\gamma_1D_i + \gamma_2G_i(t) + \gamma_3 Z_i\},
\end{equation}
where $\nu_i$ is Gamma distributed with mean 1 and variance $\sigma^2_{\eta} = 0.1$, which makes the observation times positively correlated. We let $\lambda_0(t) = \frac{\sqrt{t}}{2}$, and consider various values for $\boldsymbol{\gamma} = (\gamma_1, \gamma_2, \gamma_3)$, which will vary across the simulations. To simulate the observation times, we turn to a thinning method presented by Lewis and Shedler \cite{Lewis1979}. The thinning method is one of the most popular methods for generating a nonhomogeneous Poisson process (NHPP) and is based on finding a constant intensity function $\bar{\lambda}$ for a homogeneous Poisson process (HPP) that dominates the desired intensity $\lambda(t)$. A rejection-acceptance algorithm is used to reject a portion of the generated observation times until the desired intensity or rate is achieved. Algorithm \ref{alg:thinning} shows the thinning algorithm for an intensity function $\lambda(t)$ over $(0, \tau]$. We also generate the censoring time as $C_i \sim Unif(\tau/2, \tau)$ where $\tau$ is the study end time, unless otherwise indicated. We let $\tau = 7$ in all simulations. All computations are performed using R version 4.2.2 \cite{Rsoftware}. Code for all of the simulation studies are publicly available on GitHub at \href{https://github.com/grcetmpk/MIPTIW}{https://github.com/grcetmpk/FIPTIW}.

\begin{algorithm}
    \KwInput{Intensity function $\lambda(t)$, maximum follow-up $\tau$}
    Initialize $n = 0, m = 0, t_0 = 0, s_0 = 0, \bar{\lambda} = \sup_{0\le t \le \tau}\lambda(t)$\;
    \While{$s_m < \tau$}
    {
    Generate $u \sim Unif(0, 1)$\;
    Let $b = -\log(u)/\bar{\lambda}$\;
    Set $s_{m+1} \gets s_m + b$\;
    Generate $R \sim Unif(0,1)$\;
    \If{$R \le \lambda(s_{m+1})/\bar{\lambda}$}
        {
        $t_{n+1} \gets s_{m+1}$\;
        $n \gets n+1$\;
        }
    $m \gets m + 1$;\
    }
    \eIf{$t_n \le \tau$}{
            \Return{$\{t_k\}_{k = 1,2,\dots,n}$}\;
    }{
    \Return{$\{t_k\}_{k = 1,2,\dots,n-1}$}\;
    }
\caption{Thinning Algorithm}   
\label{alg:thinning}
\end{algorithm}

\subsection{Simulation I: Violations of the Noninformative Censoring Assumption}\label{sec:censoringassumptionsim}

As it is known that censoring/dropout can bias model parameters when it is related to the longitudinal outcome \cite{Ma2005}, we investigate the performance of the FIPTIW-GEE under various violations of the noninformative censoring assumption in this simulation study. We also examine how further incorporating IPCW weights into the FIPTIW method affects estimation of the outcome model parameters. We denote the FIPTIW method that also incorporates IPCW weights as the FIPTICW method.

We simulate covariates as in Section \ref{sec:datagen}, and we simulate the probability of being treated as $\pi_i = \text{expit}(\alpha_0 + \alpha_1W)$ where $\alpha_0 = -1$ and $\alpha_1 = 1$. The outcome is simulated as in Equation (\ref{eq:simoutcomenormal}) and we simulate observation times according to Equation (\ref{eq:simintensity}), where we let $\gamma_1 = 0.5, \gamma_2 = 0.3$, and $\gamma_3 = 0.6$. 

In this simulation, we allow the covariates driving the observation times to also be related to the censoring time. We specify the censoring hazard as
\begin{equation}\label{eq:simIcenshazard}
    \lambda_{i,c}(t) = \lambda_{0, c}\exp\{\eta_1D_i + \eta_2W_i + \eta_3Z_i\},
\end{equation}
where $\lambda_{0, c} = 0.1 \times t$ is the baseline censoring hazard. In this simulation, we let $\eta_1 = 0.4$, $\eta_2 = (0, 0.2, 0.5)$, and $\eta_3 = (0, 0.4, 0.6)$ to see how the strength of the relationship between various covariates and the censoring times affects the estimation of the ATE, $\beta_1$, in Equation (\ref{eq:simulationmarginalmodel}). Under this hazard function, when $\eta_2 = \eta_3 = 0$, the noninformative censoring assumption is satisfied. Otherwise, the assumption is violated. 

To simulate the censoring time, we can use the inverse probability method proposed by Bender \cite{Bender2005}. We can use the PH model in Equation (\ref{eq:simIcenshazard}) to construct a survival function as 
\begin{equation*}
    S(t) = \exp\{H_0(t)\times \exp(\eta_1D_i + \eta_2W_i + \eta_3Z_i)\},
\end{equation*}
where $H_0(t) = \int_{0}^t \lambda_{0, c}(u)du$. So long as $\lambda_{0, c} > 0$ at all $t$, then  the survival/censoring time can be expressed as 
\begin{equation*}
C_i = H_0^{-1}\left[-\log(U) \times \exp(-(\eta_1D_i + \eta_2W_i + \eta_3Z_i))\right],
\end{equation*}
where $U \sim Unif(0, 1)$ \cite{Bender2005}. As such, we simulate the censoring time for each individual by randomly drawing $U_i$ from a uniform distribution on (0,1) and calculate
\begin{equation*}
    C_i = \sqrt{\frac{2}{0.1}(-\log(U_i)\times\exp(-\eta_1D_i + -\eta_2W_i + -\eta_3Z_i))}
\end{equation*}
for each individual. 

Under this data generating mechanism, Assumptions \ref{O1}, \ref{O3}, \ref{O4}, and \ref{O5} hold for the observation process, and \ref{T1} to \ref{T5} hold for the treatment assignment process. Assumption \ref{O2} will only hold when  $\eta_2$ and $\eta_3$ are both zero. 

For each combination of $(\eta_1, \eta_2, \eta_3)$, we simulate 1000 data sets for varying sample sizes of $n$ = 50, 100, and 500. For each data set, we estimate the ATE, $\beta_1$, from an unweighted independent GEE and independent GEEs weighted by IPTW, IIW, FIPTIW, and FIPTIW including the IPCW weights (FIPTICW), where the IPCW, IIW, and IPTW weights were multiplied together. The IIW weights (stabilized) are estimated using a Cox PH model, the IPTW weights are estimated by logistic regression, and the IPCW weights are estimated by a second Cox PH model. From each of these models over the 1000 simulated data sets, we calculate the bias and mean squared error (MSE) under each simulation scheme. The IPCW weights are estimated using a Cox PH model. All models are correctly specified. 

The results for $n = 100$ are shown in Figure \ref{fig:censoring_n100}. The results for $n = 50$ and $n = 500$ can be found in Supplementary Material \ref{sec:appendixsimI}, and are summarized below.

\begin{figure}[ht]
    \centering
    \includegraphics[width = \textwidth]{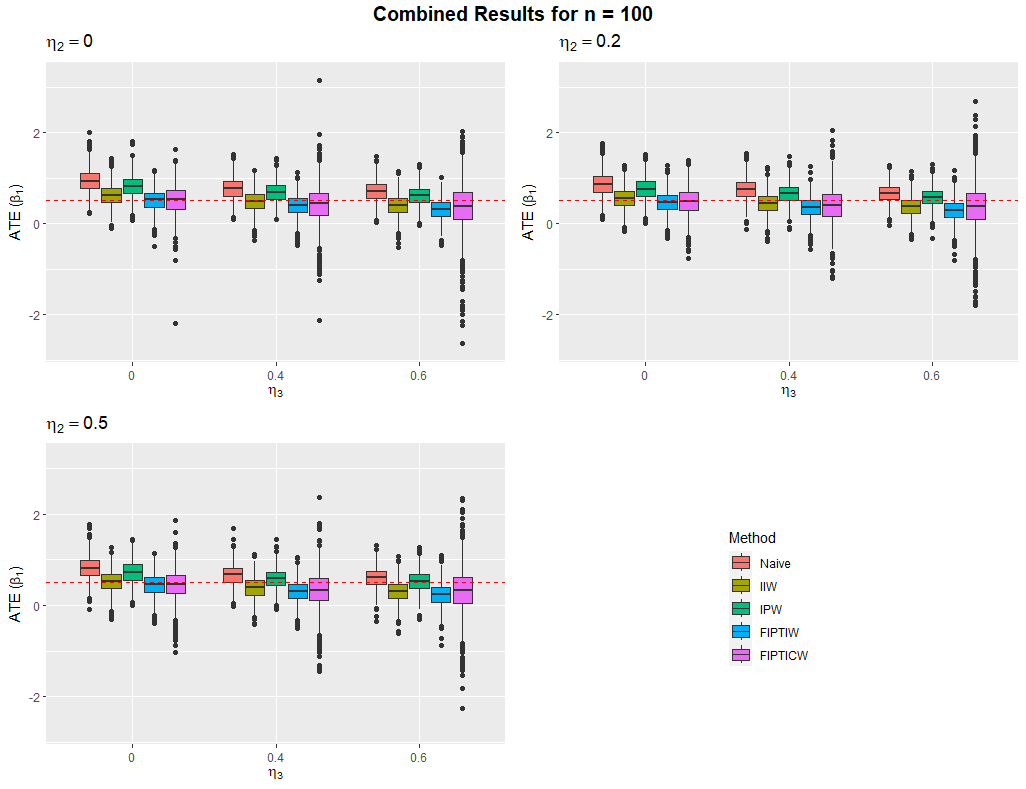}
    \caption{Results of Simulation I for $n$ = 100. Average treatment effect (ATE, $\beta_1$) is calculated by fitting an independent GEE with various weights for each simulation scheme over the 1000 generated data sets. The weighting methods include no weighting (unweighted), inverse intensity weighting (IIW), inverse probability of treatment weighting (IPTW), flexible inverse probability of treatment and intensity weighting (FIPTIW) and flexible inverse probability of treatment weighting with inverse probability of censoring weights included (FIPTICW). The true value of the ATE is 0.5, and is denoted by the red horizontal line.}
    \label{fig:censoring_n100}
\end{figure}

When $n = 100$ and $\eta_1, \eta_2$ and $\eta_3$ are all zero, the noninformative censoring assumption (Assumption \ref{O2}) is satisfied. In this setting, the FIPTIW and FIPTICW methods are unbiased while the other methods have large biases. The FIPTICW method, though showing the smallest bias, has a large variance. When $\eta_1$ is non-zero while $\eta_2$ and $\eta_3$ are zero, the FIPTIW and FIPTICW methods remain unbiased.  However, the FIPTICW method again has large variance. This shows that when the treatment indicator (which is included in the outcome model) is related the censoring hazard, the FIPTIW method can be used to estimate the ATE. 

As the magnitude of $\eta_3$ increases, more bias is introduced into our estimated ATEs under the FIPTIW and FIPTICW methods. In most cases, the bias and variance increase as the the magnitude of $\eta_3$ increases. Interestingly, the bias of the other methods do not tend to increase as $\eta_3$ increases. In fact, we often see the bias and MSE decrease as $\eta_3$ increases under the unweighted, IIW, and IPTW methods. As $\eta_2$ increases, we see the bias and variance also tend to increase for the FIPTIW and FIPTICW methods. Again, the bias decreases as $\eta_2$ increases for the unweighted, IIW, and IPTW methods in many cases.

There are many settings where the FIPTIW and FIPTICW methods had worse performance than the IIW and IPTW methods in terms of both bias and variance. Aside from when $\eta_2$ and $\eta_3$ are zero, the IIW and/or the IPTW methods outperform the FIPTIW and FIPTICW methods in most settings. In some cases, the unweighted method outperforms the FIPTIW and FIPTICW methods. These results are surprising as the IIW and IPTW methods only account for one of the three sources of bias in the data. The unweighted method accounts for none of the three sources of bias. We also note that often the best performing method in terms of bias and variance still results in biased estimation. 

In general, the FIPTICW method tends to have a smaller bias than the FIPTIW method. However, it also tends to have a higher variance and is still biased in many settings as the IIW weights are not adequately estimated. 

We see similar results for $n = 50$ (Figure \ref{fig:censoring_n50} in Supplementary Material \ref{sec:appendixsimI}) with larger variances for the estimated ATEs. We again see that the FIPTIW and FIPTICW methods are often outperformed by the IIW and IPTW methods, and provide biased results for the estimated ATEs. When $n = 500$ (Figure \ref{fig:censoring_n500} in Supplementary Material \ref{sec:appendixsimI}), we see a decrease in the variances of the estimated ATEs. Again, the FIPTIW and FIPTICW methods are often biased and outperformed by the IIW and IPTW methods. 

From these results, we conclude that there is moderate sensitivity to the noninformative censoring assumption for the FIPTIW method. That is, the FIPTIW may result in spurious estimates of the ATE when the noninformative censoring assumption is not satisfied. Further including the IPCW weights into the FIPTIW model does not fully adjust for the bias introduced by the noninformative censoring assumption, not to mention it results in increased variance in most settings. As such, work is needed to extend existing methodology to allow for the consideration of informative censoring in the analysis.

\subsection{Simulation II: Variable Inclusion in Inverse Intensity Weighting Models}\label{sec:varinclusion}

To examine which variables should be included in IIW models, we perform a simulation study similar to Brookhart et al. \cite{Brookhart2006}.

In this simulation, we simulate the covariates as in  Section \ref{sec:datagen} where the probability of being treated is $\pi = 0.5$. That is, the treatment assignment is randomized with equal likelihood of treatment allocation. We choose to simulate treatment assignments that are completely randomized to focus only on the variable inclusion for IIW, as variable inclusion for IPTW has previously been investigated \cite{Brookhart2006, Zhu2015}. We simulate the outcome as in Equation (\ref{eq:simoutcomenormal}). We simulate the observation times using the intensity model in Equation (\ref{eq:simintensity}) where we let $\gamma_1 = 0.5$, $\gamma_2 = \{0, 0.3\}$, and $\gamma_3 = 0.6$. Under this mechanism, assumptions \ref{O1}, \ref{O2}, \ref{O3}, and \ref{O5} hold. Assumption \ref{O4} will hold if all covariates are included in the intensity model (i.e. if the model is correctly specified).  As the treatment assignment was randomized with equal likelihood, \ref{T1} to \ref{T5} implicitly hold. 

Under this data generating mechanism, we only need to fit an IIW model to adjust for the conditionally ignorable observation process as the treatment assignment is fully randomized. The covariate $Z$ is related to the observation intensity, treatment, and the longitudinal outcome, making it a confounder for the estimation of the ATE. When $\beta_2 = 0$ and $\gamma_2 = 0$, the covariate $G(t)$ is not related to the observation intensity or longitudinal outcome. When $\beta_2 = 2$ and $\gamma_2 = 0$, the covariate $G(t)$ is only related to the longitudinal outcome. When $\beta_2 = 0$ and $\gamma_2 = 0.3$, the covariate $G(t)$ is only related to the observation intensity. When $\beta_2 = 2$ and $\gamma_2 = 0.3$, $G(t)$ is related to both processes but is not a confounder as it is not related to the treatment assignment. 

We consider sample sizes of $n = 50, 100,$ and $500$. For each sample size, we generate 1000 data sets and fit seven models for the observation intensity using all possible combinations of $D$, $G(t),$ and $Z$ as covariates. For each of these models, we obtain estimated stabilized IIW weights by fitting a Cox PH model. We then use these estimated weights in an independent GEE to obtain an estimate of $\beta_1$, the ATE. We also fit an unweighted model for comparison. From the 1000 data sets, we calculate the bias, variance, and mean squared error (MSE) of $\beta_1$ for each possible combination of covariates in the weighting model. The results for $n = 100$ are shown in Table \ref{tab:simII_n100}. 

\begin{table}[ht]
    \centering
    \begin{tabular}{ll|lllllllll}
\multicolumn{2}{c}{} & \multicolumn{9}{c}{\textbf{Variables used to estimate intensity}} \\
$\boldsymbol{\gamma_2}$ & $\boldsymbol{\beta_2}$ &  & \textbf{Naive} & $\boldsymbol{D}$ & $\boldsymbol{G(t)}$ & $\boldsymbol{Z}$ & $\boldsymbol{D,G(t)}$ & $\boldsymbol{D,Z}$ & $\boldsymbol{G(t),Z}$ & $\boldsymbol{D,G(t),Z}$\\
\hline
0 & 0 &  &  &  &  &  &  &  &  & \\
 &  & Bias: & 0.286 & 0.286 & 0.285 & 0.082 & 0.285 & 0.026 & 0.082 & 0.026\\
 &  & MSE: & 0.102 & 0.102 & 0.102 & 0.024 & 0.102 & 0.017 & 0.024 & 0.017\\
0 & 2 &  &  &  &  &  &  &  &  & \\
 &  & Bias: & 0.287 & 0.287 & 0.289 & 0.085 & 0.286 & 0.029 & 0.083 & 0.029\\
 &  & MSE: & 0.112 & 0.112 & 0.116 & 0.032 & 0.111 & 0.026 & 0.031 & 0.026\\
0.3 & 0 &  &  &  &  &  &  &  &  & \\
 &  & Bias: & 0.299 & 0.299 & 0.297 & 0.090 & 0.297 & 0.033 & 0.089 & 0.032\\
 &  & MSE: & 0.110 & 0.110 & 0.108 & 0.025 & 0.108 & 0.018 & 0.025 & 0.018\\
0.3 & 2 &  &  &  &  &  &  &  &  & \\
 &  & Bias: & 0.383 & 0.383 & 0.318 & 0.169 & 0.315 & 0.111 & 0.100 & 0.043\\
 &  & MSE: & 0.181 & 0.181 & 0.137 & 0.057 & 0.132 & 0.040 & 0.036 & 0.028\\
\end{tabular}
    \caption{Simulation results for Simulation II for $n = 100$. Bias and mean squared error (MSE) of the average treatment effect (ATE) is calculated by weighting the outcome model in Equation (\ref{eq:simulationmarginalmodel}) by inverse intensity weighting (IIW) for each simulation scheme over the 1000 generated data sets. Variables included in the IIW model are listed in the table. The true value of the ATE is 0.5.}
    \label{tab:simII_n100}
\end{table}

For a sample size of 100, when $\gamma_2 = 0$ and $\beta_2 = 0$, the covariate $G(t)$ is not related to the observation intensity or longitudinal outcome. In this setting, any weighting model that does not include the confounder $Z$ and treatment $D$ are biased. Further including an unrelated covariate $G(t)$ does not impact the bias, variance, and thus the MSE in this setting.

When $\gamma_2 = 0$ and $\beta_2 = 2$, the covariate $G(t)$ is related only to the longitudinal outcome. Again, any weighting model that does not include the confounder $Z$ and treatment $D$ are biased. Further including the covariate only related to the longitudinal outcome, $G(t)$, did not influence the bias or MSE of the ATE. Similar results are seen when $\gamma_2 = 0.3$ and  $\beta_2 = 0$ (the covariate $G(t)$ is only related to the observation intensity). 

When both $\gamma_2$ and $\beta_2$ are non-zero, $G(t)$ is related to both the observation and outcome processes but is unrelated to the treatment assignment. In this setting, the estimate is only unbiased when the weighting model includes all three covariates.

The results for $n = 50$ and $n = 500$ can be found in Tables \ref{tab:simII_n50} and \ref{tab:simII_n500} in Supplementary Material \ref{sec:appendixsimII}. 
For the case of $n = 50$, we observe similar trends and see larger variances of the parameter estimates. Similarly, for $n = 500$ we observe similar trends and see smaller variances of the parameter estimates. 

The results of this simulation highlight the importance of including observation process confounders in the intensity model. Including covariates that do not confound the relationship between the observation times and the outcome did not impact estimation greatly, with no substantial increases in MSE. Specifically, we do not see an increase in variance when including variables that are only related to the observation intensity, which is contrary to what has been shown in the IPTW literature. As such, it is recommended to be conservative (i.e., more inclusive) with the variables included in the intensity model if there is reason to believe they could potentially be related to both the outcome and observation times.

\subsection{Simulation III: Weight Trimming}\label{sec:weighttrimming}

In this simulation, we consider weight trimming under FIPTIW when the IPTW and IIW weights are extreme due to the underlying observation intensity and/or propensity score. We recognize extreme weights can also occur under model misspecification, but limit this analysis to the scenario where the observation intensity and propensity score models are correctly specified.  

We simulate data as in Section \ref{sec:datagen} where the probability of being treated is $\pi = \text{expit}(\alpha_0 + \alpha_1W)$ and $\alpha_0 = -1$ and $\alpha_1$ will vary for each simulation scheme. In this simulation, we generate the outcome as in Equation (\ref{eq:simoutcomenormal}). We simulate the observation times according to the intensity in Equation (\ref{eq:simintensity}), where the parameters $\boldsymbol{\gamma}$ will also vary for each scenario. Under this data generating mechanism, Assumptions \ref{O1} to \ref{O5} and \ref{T1} to \ref{T5} all hold. That is, both the observation and treatment assignment processes are conditionally ignorable. 

We consider scenarios where we have varying degrees of informativeness in the treatment assignment and observation processes. To simulate various strengths of informativeness in the propensity score model, we set $(\alpha_0, \alpha_1)$ to $(0, 0.5), (0, 3.5),$ or $(0, 5.5)$ which correspond to low, moderate, and high degrees of informativeness. To simulate various observation processes, we set $(\gamma_1, \gamma_2, \gamma_3)$ to $(0.5, 0.3, 0.6), (0.5, 0.3, -0.75)$, or $(0.5, 0.3, -1.1)$, which correspond to low, moderate, and high degrees of informativeness. The degrees of informativeness in each process will individually and jointly affect the resultant IPTW, IIW, and FIPTIW weights. We consider the following combinations of the degrees of informativeness for the treatment assignment/observation processes: low/low, moderate/low, high/low, low/moderate, low/high, and moderate/moderate. Under the moderate/high, high/moderate, and high/high scenarios, we end up with propensity score estimates of 0 or 1 which causes the estimation of the ATE to break down. We omit the results of such analyses.

For each simulation scheme, we simulate 1000 data sets with $n = 100$ and calculate the stabilized IIW weights and IPTW weights through a Cox PH and logistic regression model, respectively. We obtain various estimates of the ATE ($\beta_1$) by trimming weights to the $p$th percentile from $p = 0.50$ to $1.00$ in increments of 0.01 for each data set. We consider trimming the individual IIW and IPTW weights prior to multiplying (which we refer to as ``trimming first") and also trimming the FIPTIW weights after multiplying the IIW and IPTW weights together (which we refer to as ``trimming after"). We then aggregate the results and plot the relative bias (RB), variance, and MSE of the estimated ATE over the 1000 data sets for each cut point and trimming method considered.

\begin{table}[h]
\centering
\begin{tabular}{@{}ll|rrr|rrr|rrr@{}}
\multicolumn{2}{c|}{\textbf{\begin{tabular}[c]{@{}c@{}}Degree of \\ Informativeness\end{tabular}}} & \multicolumn{3}{c|}{\textbf{\begin{tabular}[c]{@{}c@{}}Mean Proportion \\ of IPTW Weights\end{tabular}}} & \multicolumn{3}{c|}{\textbf{\begin{tabular}[c]{@{}c@{}}Mean Proportion \\ of IIW Weights\end{tabular}}} & \multicolumn{3}{c}{\textbf{\begin{tabular}[c]{@{}c@{}}Mean Proportion \\ of FIPTIW Weights\end{tabular}}} \\
\hline
\begin{tabular}[c]{@{}l@{}}Treatment \\ Assignment \\ Process\end{tabular} & \begin{tabular}[c]{@{}l@{}}Observation \\ Process\end{tabular} & \textgreater 5 & \textgreater 10 & \textgreater 20 & \textgreater 5 & \textgreater 10 & \textgreater 20 & \textgreater 5 & \textgreater 10 & \textgreater 20 \\ 
\hline
Low & Low & 0.00 & 0.00 & 0.00 & 0.00 & 0.00 & 0.00 & 0.00 & 0.00 & 0.00 \\
Moderate & Low & 0.36 & 0.09 & 0.02 & 0.00 & 0.00 & 0.00 & 0.36 & 0.09 & 0.02 \\
High & Low & 0.61 & 0.22 & 0.08 & 0.00 & 0.00 & 0.00 & 0.61 & 0.22 & 0.08 \\
Low & Moderate & 12.67 & 8.01 & 3.99 & 1.89 & 0.07 & 0.00 & 12.67 & 8.01 & 3.99 \\
Low & High & 13.78 & 6.83 & 4.40 & 6.17 & 0.63 & 0.04 & 13.78 & 6.83 & 4.40 \\
Moderate & Moderate & 12.53 & 3.38 & 1.79 & 10.26 & 1.39 & 0.10 & 12.53 & 3.38 & 1.79\\
\\
\end{tabular} \caption{Distribution of weights under each simulation scheme.}
    \label{tab:trimming_weights}
\end{table}

Table \ref{tab:trimming_weights} shows the average proportion of the estimated weights larger than 5, 10, and 20 prior to trimming across the 1000 simulated data sets. When the observation process is simulated to have a low degree of informativeness, we see the mean proportion of IPTW and FIPTIW weights larger than 5, 10, and 20 monotonically increase as the treatment assignment process becomes more informative. The IIW weights are stable, with a negligible mean proportion of weights larger than 5, 10, and 20. However, when the treatment assignment process is simulated to have a low degree of informativeness and the observation process has a moderate or high degree of informativeness, we see a large mean proportion of extreme IPTW weights despite the underlying treatment assignment process. We see extreme IIW weights to a much lesser extent in these scenarios, despite simulating highly informative observation processes. This result highlights how one process may effect the estimation of weights for another. 

Figure \ref{fig:trimmingresults} shows the results of the estimated ATE for each of the six scenarios considered. Detailed descriptions of the results for each individual simulation scheme can be found in Supplementary Material \ref{sec:appendixsimIII}. The results show that weight trimming is beneficial when extreme IPTW weights are produced based on the underlying treatment assignment process. That is, we found that weight trimming reduced the relative bias of our estimated ATEs when the treatment assignment process had moderate or high degrees of informativeness. Interestingly, when extreme IPTW and IIW weights were produced because of the informativeness in the observation process, weight trimming did not improve the performance of the FIPTIW method. However, in this setting trimming weights above the 95th percentile do not significantly change the results. That is, weight trimming can reduce the overall bias introduced into the estimated ATEs when the treatment assignment process is informative, without a significant change in MSE. Even when weight trimming is not optimal, the difference in bias and MSE between trimming at high percentiles versus not trimming at all are negligible. As such, we recommend trimming FIPTIW weights to the 95th percentile when extreme weights are present. The results showed that at this percentile, the difference in estimation when trimming FIPTIW weights prior to multiplying compared to trimming after multiplying are negligible. Therefore, one can choose to trim weights either before or after multiplying the IIW and IPTW weights together when employing FIPTIW.  

\begin{figure}[ht]
\centering
\includegraphics[width=\textwidth, height = 17cm]{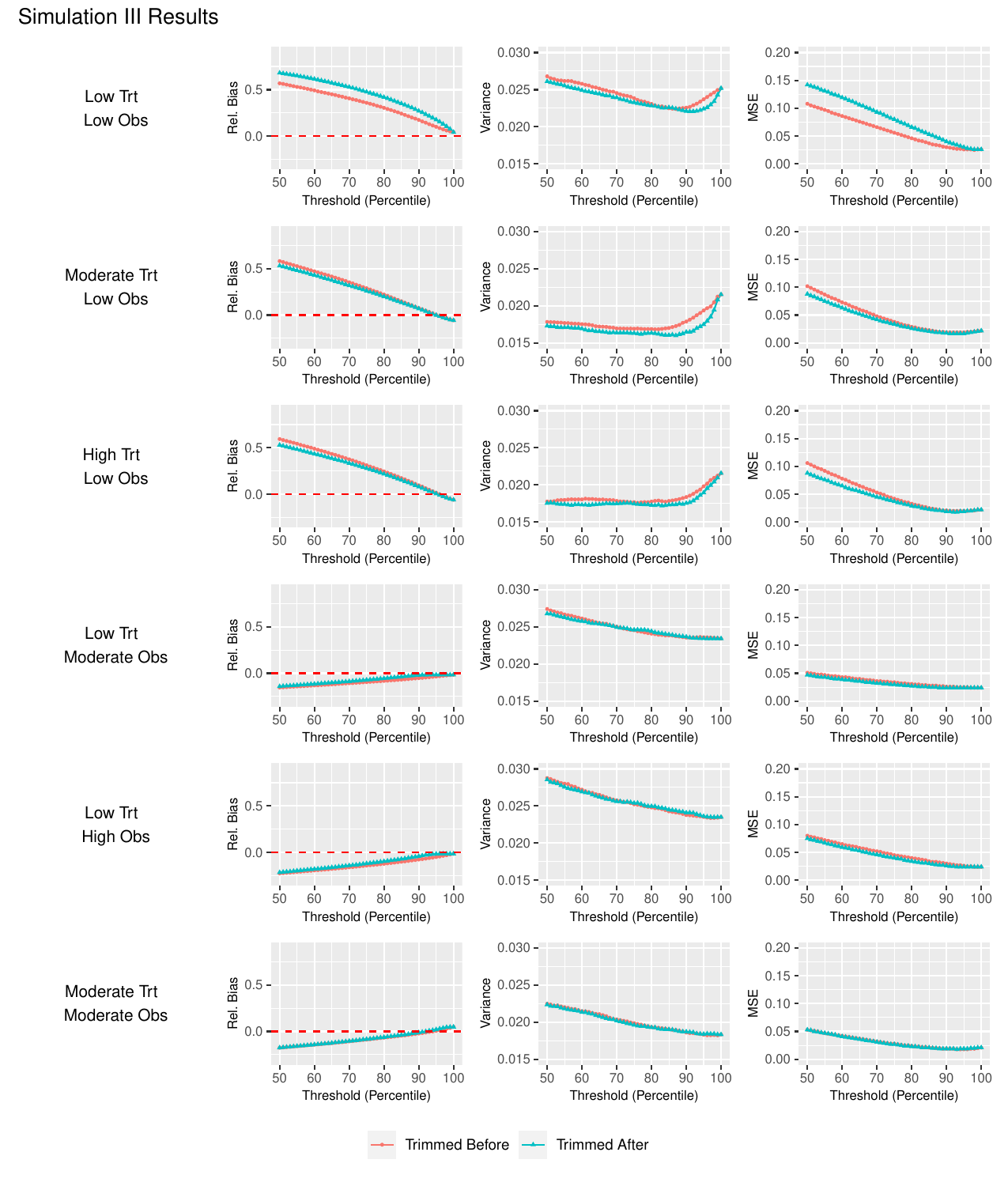}\hfill
\caption{Results of the simulation under various degrees of informativeness (low, moderate, and high) in the treatment assignment (``Trt") and observation (``Obs") processes.  ``Trimmed before" refers to trimming the individual IPTW and IIW weights prior to multiplying and calculating the FIPTIW weights. ``Trimmed after" refers to trimming the FIPTIW weights after multiplying the IPTW and IIW weights together. The true average treatment effect $\beta_1$ is simulated to be 0.5.}
\label{fig:trimmingresults}
\end{figure}

\section{Application on the PRISM Study}\label{sec:FIPTIWdatanalysis}

Malaria is spread through infected female mosquitoes, which require aquatic habitats to breed \cite{shayo2021}.  Larval source management has been used as a malaria prevention strategy to reduce both the densities of infected adult mosquitoes and prevalence of malaria in humans \cite{msellemu2016epidemiology, mwakalinga2018topographic}. A recent study has shown an association between residing in households with non-piped water sources and malaria diagnoses in a sample of 6,707 children in Tarzania. Further, a recent meta-analysis found that unprotected household water sources such as unprotected wells, springs, rivers, dams, and streams were associated with higher prevalence of malaria in young children \cite{yang2020drinking}. We similarly wish to evaluate if there is an increase in the odds of malaria diagnosis for those residing in households with unprotected water sources compared to those with non-piped, protected water sources among children ages 2 to 11 living in Uganda. 

We analyze observational data from the Program for Resistance, Immunology, Surveillance, and Modeling of Malaria in Uganda (PRISM) \cite{kamya2015malaria} using the FIPTIW method to quantify the association between unprotected water sources and malaria diagnoses. In the PRISM study, households in three sub-counties of Uganda were randomly selected to be enrolled in a longitudinal study. All individuals in the household were enrolled in the study if they met the eligibility criteria outlined by Kamya et al. \cite{kamya2015malaria}. Data were collected on individuals enrolled in the study between August 2011 and June 2017 through routine clinical visits roughly every three months. Participants also attended a study clinic (and thus had observations recorded) any time they became ill.

The PRISM study provided demographic information on each study participant, including the individual's age and sex at enrollment, and the individual's weight, height, body temperature, malaria diagnosis, and prescribed treatments at each scheduled and unscheduled assessment. The study also provided information on the household that each patient resides in, including the sub-county in Uganda, dwelling characteristics such as materials, number of sleeping places, waste disposal, household water sources, a categorical and numerical household wealth index, and food security measures such as the number of meals given per day or the number of food problems per week. Household characteristics were matched to the participant through a unique household identifier. Individuals were classified as residing in households with protected water sources if their dwelling sourced drinking water from a protected spring or well, public tap, or borehole. Individuals were classified as having unprotected water sources if their drinking water was sourced from open public wells, rivers and streams,  ponds and lakes, or unprotected springs. This categorization of drinking water sources closely follows that of Yang et al. \cite{yang2020drinking}. We note that drinking water source (along with other household characteristics) did not change during follow-up and are considered time-invariant covariates.  Participants that had drinking water piped into their yard, compound, or dwelling were excluded from the study to focus on the difference within unpiped water sources.

Individuals residing within the same household are likely to be related on biological and environmental factors. To restore the independence assumption between individuals in the study, we randomly selected one child per household to be in the final sample. The final sample size was 287 children, all of whom resided in separate houses. From this sample, 218 children resided in households with protected water sources while 69 resided in households with unprotected water sources. Baseline demographics for other covariates considered in the study are given in Table \ref{tab:malariademog1}.

\begin{table}[h]
    \centering
\begin{tabular}{ll|ll}
\textbf{Covariate }& \textbf{Mean (SD) or n (\%)} & \textbf{Covariate }& \textbf{Mean (SD) or n (\%)}\\
\hline
&&&\\
\textbf{Age at Enrollment} & 5.51 (2.03) & \multicolumn{2}{l}{\textbf{Unprotected Water Source}}\\
\multicolumn{2}{l|}{\textbf{Sex}} & \hspace{1em}No & 218 (75.96\%)\\
\hspace{1em}Female & 145 (50.52\%)& \hspace{1em}Yes & 69 (24.04\%)\\
\hspace{1em}Male & 142 (49.48\%)&\multicolumn{2}{l}{\textbf{Dwelling Type}}\\
\multicolumn{2}{l|}{\textbf{Sub-County}}&\hspace{1em}Modern & 85 (29.62\%)\\
\hspace{1em}Walukuba & 85 (29.62\%)&\hspace{1em}Traditional & 202 (70.38\%)\\
\hspace{1em}Kihihi & 100 (34.84\%)&\multicolumn{2}{l}{\textbf{Food Problems per Week}}\\
\hspace{1em}Nagongera & 102 (35.54\%)&\hspace{1em}Sometimes & 97 (33.8\%)\\
\multicolumn{2}{l|}{\textbf{Household Wealth Index}}&\hspace{1em}Never & 47 (16.38\%)\\
\hspace{1em}Least poor & 86 (29.97\%)&\hspace{1em}Often & 38 (13.24\%)\\
\hspace{1em}Poorest & 103 (35.89\%)&\hspace{1em}Always & 55 (19.16\%)\\
\hspace{1em}Middle & 98 (34.15\%)&\hspace{1em}Seldom & 50 (17.42\%)\\
\multicolumn{2}{l|}{\textbf{Drinking Water Source}}&\multicolumn{2}{l}{\textbf{Waste Facilities}}\\
\hspace{1em}Public tap & 105 (36.59\%)&\hspace{1em}{}Covered pit latrine, no slab & 103 (35.89\%)\\
\hspace{1em}Protected public well & 24 (8.36\%)&\hspace{1em}{}Covered pit latrine with slab & 18 (6.27\%)\\
\hspace{1em}Protected spring & 26 (9.06\%)&\hspace{1em}{}Composting  toilet & 4 (1.39\%)\\
\hspace{1em}Borehole & 63 (21.95\%)&\hspace{1em}{}Uncovered pit latrine, no slab & 112 (39.02\%)\\
\hspace{1em}River/stream & 25 (8.71\%)&\hspace{1em}{}Flush toilet & 8 (2.79\%)\\
\hspace{1em}Open public well & 21 (7.32\%)&\hspace{1em}{}Uncovered pit latrine with slab & 8 (2.79\%)\\
\hspace{1em}Pond/lake & 5 (1.74\%)&\hspace{1em}{}Vip latrine & 3 (1.05\%)\\
\hspace{1em}Unprotected spring & 18 (6.27\%)&\hspace{1em}{}No facility & 31 (10.8\%)\\
\textbf{Number of Persons Living in House} & 6.18 (2.6) & \\
\hspace{1cm}
\end{tabular}
    \caption{Baseline demographics for the 287 children included in the analysis.}
    \label{tab:malariademog1}
\end{table}

As the FIPTIW method has been shown to be sensitive to violations of the noninformative censoring assumption, we employed artificial censoring \cite{Robins1995, robins1992g, joffe2001administrative} to avoid such violations. Data on the reason for study withdrawal, but not the dates of withdrawal, were provided in the data. As such, the true censoring times were unknown. As a proxy for the true censoring time, we considered the last known observation time as the \emph{surrogate censoring time} for each individual. The surrogate censoring time is used to estimate the unknown true censoring time, which are the earliest possible times at which an individual could be right-censored. Using this definition, surrogate censoring times varied from 0 (baseline) to 1,781 days (study completion), and 5.65\% of individuals were potentially censored prior to 6 months. Of this set of individuals, 26.83\% were considered to violate the noninformative censoring assumption (i.e., dropped out of the study due to reaching 11 years of age, withdrawing informed consent, being unable to comply with the study protocol, or being unable to be located for more than 120 days). As we used the smallest possible censoring time for each individual, this means that a maximum of 1.52\% of the the total study population could be considered censored prior to 6 months for reasons possibly related to the longitudinal outcome. We also performed a sensitivity analysis where individuals who had surrogate censoring times prior to 6 months were randomly censored between 0 and 6 months (see Supplementary Material \ref{sec:appendixSArandom}). There were negligible differences in the results of the analysis when random censoring was employed for this small proportion of individuals. As such, we consider this proportion of potentially informative censoring negligible. For this analysis, we use 6 months (182.5 days) as the point at which all individuals are artificially censored to restore the noninformative censoring assumption.

As patients could visit clinics between scheduled follow-ups, the number of observations per participant varied. From the 133 children we included in this sample, the number of observations ranged from 1 to 11, with an average of 4.83 observations per participant. Of the 1,386 total observations, 287 were for enrollment, 525 were scheduled follow-up appointments, and 574 were unscheduled appointments. Of the 1,386 visits, Malaria was diagnosed 179 times (12.9\%), and 111 of 287 patients (38.7\%) had a Malaria diagnosis at some point during the study. 

We wish to estimate the average treatment effect of a time-invariant indicator of whether or not an individual resides in a household with a non-piped water source on malaria diagnoses. That is, we wish to estimate the marginal model
\begin{equation}\label{eq:malariaoutcomemod}
    g(E(Y_i(t)|D_i) = \mu(t) + \beta_1 D_i
\end{equation}
where $Y_i(t)$ is a binary, time-varying indicator of whether or not the patient was diagnosed with malaria at time $t$, $\mu(t)$ is a smooth function of time, $D_i$ is a time-invariant indicator of whether or not the individual resided in a household with unprotected water sources, and $g(\cdot)$ is the logit link function. 

To estimate the average treatment effect, we considered adjusting for factors that may have confounded the relationship between household water sources and malaria diagnoses through the propensity score. To do so, we employed a logistic regression model. We considered age, sub-county, categorical wealth index, number of food problems per week, type of household waste disposal, dwelling type, and number of people living in the house as potential confounders. We note that we also have data on whether the patient was prescribed artmether-lumefantrine, quinine, artesunate, or no malaria medication for each observation. These medications are not preventive against malaria but used to treat it. However, in our sample only no observations had antimalarial medications prescribed within one week of the current observation time, so we omit this covariate from the analysis. After an initial model fit, the sub-county and human waste facility covariates caused inflated variances in the propensity score model parameter estimates. As such, these were omitted from the final model. The resultant propensity score model was used to estimate IPTW weights for each individual at each time. Large IPTW weights were present, where 5.69\% of weights were larger than 5, and 0.29\% of weights were larger than 10. No weights were larger than 20. The maximum IPTW weight was 10.52. This is similar to what was seen in the moderate category for the IPTW weights we defined in the simulation study in Section \ref{sec:weighttrimming}. 

We also considered the factors driving the observation process. Based on the results of Simulation II (Section \ref{sec:appendixsimII}), we included covariates believed to be related to both the probability of being observed and the longitudinal outcome. As such, we allowed the probability of visiting a clinic at any given time to depend on whether or not the individual resided in a household with protected water sources, their age, sub-county, dwelling type, categorical household wealth index, degree of food problems per week, number of persons living in the house, and the individual's malaria status at the last clinic visit. A Cox PH model with these covariates was used to estimate stabilized IIW weights. From this model, the maximum weight was 3.66. This is similar to the low informativeness scenario as defined in Simulation II.

The resultant FIPTIW weights were calculated by multiplying the IPTW and IIW weights together for each individual at each time point. The FIPTIW weights had some extreme values, where 8.44\% of the FIPTIW weights were above 5, 0.51\% were above 10, and the maximum FIPTIW weight was 13.19. This is similar to what was seen when the treatment assignment process was moderately informative and the observation process had low informativeness in Section \ref{sec:weighttrimming}. 

To estimate $\beta_1$ in the outcome model in Equation (\ref{eq:malariaoutcomemod}), the smooth function of time $\mu(t)$ was estimated using a cubic spline with a constant intercept, as in Coulombe et al. \cite{coulombe2021}. The knots of the spline were chosen by finding the tertiles of the time since enrollment. Independent GEEs were fit using various weighting methods (none, IPTW only, IIW only, FIPTIW, trimmed FIPTIW). When employing the trimmed FIPTIW weighting method, we trimmed the FIPTIW weights after multiplication to the 95th percentile. This resulted in FIPTIW weights above 3.83 being trimmed. The results of the analysis for each weighting method are shown in Table \ref{tab:malariaresults1}.  Two sensitivity analyses were also performed for this study. Supplementary Material \ref{sec:appendixSAcensoring} presents a sensitivity analysis where we include individuals who are censored into the analysis. Supplementary Material \ref{sec:appendixSAclustering} presents a sensitivity analysis where we include all children in each household.

\begin{table}[ht]
    \centering
\begin{tabular}{l|rrrrr}
 \textbf{ Weighting Method} & $\boldsymbol{\eta_1}$ & \textbf{SE(}$\boldsymbol{\eta_1}$\textbf{)} & \textbf{95\% CI for} 
 $\boldsymbol{\eta_1}$ & \textbf{Odds Ratio (OR)} & \textbf{95\% CI for OR}\\
 \hline
None & 0.429 & 0.186 & (0.065, 0.793) & 1.536 & (1.067, 2.211)\\
IPTW & 0.327 & 0.168 & (-0.002, 0.657) & 1.387 & (0.998, 1.928)\\
IIW & 0.555 & 0.178 & (0.206, 0.903) & 1.742 & (1.229, 2.468)\\
FIPTIW & 0.424 & 0.168 & (0.095, 0.753) & 1.528 & (1.099, 2.122)\\
FIPTIW (Trimmed) & 0.409 & 0.166 & (0.084, 0.734) & 1.505 & (1.087, 2.083)\\
\vspace{0.1cm}
\end{tabular}
    \caption{Results of the estimation of the odds ratio (OR) of malaria diagnoses for children residing in households with unprotected water sources versus those residing in households with protected water sources ($\beta_1$) by different weighting methods for an independent GEE in the PRISM cohort. Sample does not include those with piped water sources. }
    \label{tab:malariaresults1}
\end{table}

We see the results differed based on the weighting method that was employed in the analysis. Using no weighting method estimated that the odds of being diagnosed with malaria was 1.54 times higher for individuals residing in a household with unprotected water sources (95\% CI: (1.07, 2.21)). This estimate did not account for the non-randomized exposure nor the informative observation process.  Employing only IPTW accounted for the non-randomized exposure, and reduced the estimated odds ratio to 1.39, but this estimate was  insignificant (95\% CI: (1.00, 1.93)). Employing only IIW accounted for the informative observation process and resulted in a larger estimated odds ratio of 1.76 (95\% CI: 1.23, 2.47)). Weighting the independent GEE by FIPTIW accounted for both processes, and estimated the odds ratio to be 1.53 (95\% CI: (1.10, 2.12)), which was statistically significant. However, as some extremity in the FIPTIW weights was seen in the analysis, we should consider trimming methods to reduce the bias of the causal estimates, as in Simulation III. When large FIPTIW weights were trimmed above the 95th percentile, the estimated odds ratio was 1.51 (95\% CI: (1.09, 2.08). This result was also statistically significant, and resulted in a narrower confidence interval than was seen without any weighting. 

For those without piped water sources, the results of this analysis showed a significant difference in the odds of being diagnosed with malaria for households with unprotected water sources when the observation and treatment assignment processes were accounted for in the analysis. As such, we conclude from the analysis that efforts should be made to provide individuals with access to protected water sources, along with other protective measures, to reduce their individual risk of malaria.

\section{Discussion}\label{sec:ch1discussion}

The analysis of irregular longitudinal data may be complicated by non-randomized treatment assignments and/or informative observation processes, particularly in observational data sets. The FIPTIW method can be employed in certain scenarios to account for these sources of bias. However, we have shown the existing methodology is sensitive to violations of the noninformative censoring assumption, as it may result in biased estimates of causal treatment effects when informative censoring is present. Further, the inclusion of IPCW weights into the FIPTIW model does not account for the bias introduced into the outcome model parameter estimates when informative censoring is present. This is because the estimation of the IIW weights rely on the assumption of noninformative censoring. As such, we have identified estimating IIW weights under informative censoring as an important area of research. Jackson et al. \cite{jackson2014relaxing} uses multiple imputation and bootstrapping to estimate the parameters of a Cox PH model under violations of the independent censoring assumption. It may be possible to extend this method to account for dependent censoring in the estimation of IIW weights, which we identify as an area of future work.

Variable inclusion for IIW (and thus FIPTIW) was also investigated. We have shown that omitting covariates that are related to both the observation and outcome processes in the observation intensity models can result in biased estimates of outcome model parameters. Further, we have shown that unlike in propensity score models, the inclusion of additional covariates related only to the observation intensity do not significantly increase the variance of the resultant estimates of outcome model parameters. The results of the simulations have shown that analysts should be conservative (i.e., more inclusive) with the covariates included in intensity models for IIW and FIPTIW, if there is any indication of  possibly being related to both the observation process and longitudinal outcome. 

Weight trimming was shown to reduce the bias of causal estimates for FIPTIW when the treatment assignment process was highly informative and produced extreme weights. Trimming weights around the 95th percentile often reduced the overall bias of the causal estimates in the outcome model without increasing the variance of the estimates greatly. As such, we recommend employing weight trimming when extreme IPTW (and thus FIPTIW) weights are present. There was no significant change in the estimation of the outcome model parameters when trimming FIPTIW weights before or after multiplying the IPTW and IIW weights together. 

In the real data analysis, household level clustering was present. To circumvent this issue, only one individual per household was included in the final sample. To perform a more appropriate analysis where all individuals in each household are included, one could extend the method presented in Pullenayegum et al. \cite{Pullenayegum2021} to incorporate within-household correlations in the FIPTIW method. Further, extending the FIPTIW methodology to allow for informative censoring (as previously discussed) would also allow us to use the full sample from the PRISM study. Both of these items are identified as important areas of future research.


\section*{Acknowledgements}

The authors would also like to thank Lan Wen (University of Waterloo) for her insights on the nomenclature of the assumptions.

\subsection*{Funding Acknowledgement}

G.T. gratefully acknowledges the support from the Ontario Graduate Scholarship program and the Natural Sciences and Engineering Research Council of Canada (NSERC) Canadian Graduate Scholarship program. M.W. was supported by a Canadian Institute of Health Research (CIHR) Discovery Grant. J.A.D. was supported by a Natural Sciences and Engineering Research Council of Canada (NSERC) Discovery Grant.

\subsection*{Conflict of Interest}
The authors report no Conflicts of Interest.

\section*{Data Availability Statement}

The code used to produce the results in the simulation and data application sections are openly available in a GitHub repository at \href{https://github.com/grcetmpk/FIPTIW}{https://github.com/grcetmpk/FIPTIW} and \href{https://github.com/grcetmpk/MalariaFIPTIW}{https://github.com/grcetmpk/MalariaFIPTIW}. The data used for the real data analysis is publicly available at \href{https://clinepidb.org/}{https://clinepidb.org}.


\clearpage
\newpage
\setcounter{table}{0}
\setcounter{section}{0}
\renewcommand{\thetable}{S\arabic{table}}

\setcounter{figure}{0}
\renewcommand{\thefigure}{S\arabic{figure}}
\setcounter{page}{1}

\renewcommand{\thesection}{\Alph{section}}

\fontsize{20}{12}\selectfont
\begin{center}
    \textbf{Supplementary Material}
\end{center}

\vspace{1cm}

\fontsize{10}{12}\selectfont
The following document contains supplementary material for the paper ``On Flexible Inverse Probability of Treatment and Intensity Weighting: Informative Censoring, Variable Inclusion, and Weight Trimming" by Grace Tompkins, Joel A. Dubin, and Michael Wallace.

The proof of Theorem \ref{thm:unbiased} can be found in Section \ref{sec:proofs}. Section \ref{sec:appendixsim} contains additional results for the simulations performed in the main paper. The additional results for Simulations I (sensitivity to violations of the noninformative censoring assumption), II (variable inclusion), and III (weight trimming) can be found in Sections \ref{sec:appendixsimI}, \ref{sec:appendixsimII}, and \ref{sec:appendixsimIII}, respectively. Section \ref{sec:appendixSA} contains sensitivity analyses for the PRISM cohort data analysis presented in the main paper. The sensitivity analysis for violations of the noninformative censoring assumption can be found in Section \ref{sec:appendixSAcensoring}. The sensitivity analysis for violations of the independence assumption (clustering) can be found in Section \ref{sec:appendixSAclustering}. The sensitivity analysis where individuals who had surrogate censoring times prior to 6 months were randomly censored between 0 and 6 months can be found in Section \ref{sec:appendixSArandom}. 

\clearpage\newpage

\setcounter{page}{1}

\section{Proofs}\label{sec:proofs}

\begin{proof}[Proof of Theorem \ref{thm:unbiased}]
Let $w_i^{F}(t;\boldsymbol{\alpha}, \boldsymbol{\gamma}, \pi, h)$ denote the FIPTIW weight for individual $i$ at time $t$. Without loss of generality, we prove this result using stabilized weights for general function $h(\cdot)$ of covariates $\boldsymbol{X}_i(t)$ for the IIW weights and general function $\pi(\cdot)$ of covariates $\boldsymbol{W}_i(t)$ for the IPTW weights. We begin by taking the expectation and applying the law of iterative expectations as

\begin{equation*}\label{eq:FIPTIWgeeproof1}
\begin{aligned}
    &E\{ U(\boldsymbol{\beta}; \boldsymbol{\alpha}, \boldsymbol{\gamma}, \pi, h) \} \\
    &= E \Bigg\{ \sum_{i = 1}^n \int_{0}^{\tau} \boldsymbol{X}_i(t) \left\{ \frac{dg(\mu)}{d\mu}\bigg\rvert_{\mu_i(t; \boldsymbol{\beta})} \right\} ^{-1} v(\mu_i(t; \boldsymbol{\beta}))^{-1} \{Y_i(t) - \mu_i(t; \boldsymbol{\beta}) \}w_i^{F}(t;\boldsymbol{\alpha}, \boldsymbol{\gamma}, \pi, h)dN_i(t) \Bigg\} \\
    &= E \Bigg\{ E\Bigg\{ \sum_{i = 1}^n \int_{0}^{\tau} \boldsymbol{X}_i(t) \left\{ \frac{dg(\mu)}{d\mu}\bigg\rvert_{\mu_i(t; \boldsymbol{\beta})} \right\} ^{-1} v(\mu_i(t; \boldsymbol{\beta}))^{-1}\{Y_i(t) - \mu_i(t; \boldsymbol{\beta}) \} \\
    &\qquad \qquad \times w_i^{F}(t;\boldsymbol{\alpha}, \boldsymbol{\gamma}, \pi, h)dN_i(t) \big\rvert \boldsymbol{X}_i(t) \Bigg \}  \Bigg\}\\
     &= E \Bigg\{ \sum_{i = 1}^n \int_{0}^{\tau} \boldsymbol{X}_i(t) \left\{ \frac{dg(\mu)}{d\mu}\bigg\rvert_{\mu_i(t; \boldsymbol{\beta})} \right\} ^{-1} v(\mu_i(t; \boldsymbol{\beta}))^{-1} \\
     &\qquad \qquad \times E \Big\{ \{Y_i(t) - \mu_i(t; \boldsymbol{\beta}) \}w_i^{F}(t;\boldsymbol{\alpha}, \boldsymbol{\gamma}, \pi, h)dN_i(t) \big\rvert \boldsymbol{X}_i(t)\Big\} \Bigg\}
\end{aligned}
\end{equation*}

To show this expectation is zero, we can show that the conditional expectation nested inside of the full expression is zero. We show this by again using the law of iterative expectations to re-write the conditional expectation as

\begin{equation*}
\begin{aligned}
    &E \Big\{ \{Y_i(t) - \mu_i(t; \boldsymbol{\beta}) \}w_i^{F}(t;\boldsymbol{\alpha}, \boldsymbol{\gamma}, \pi, h)dN_i(t) \big\rvert \boldsymbol{X}_i(t)\Big\} \\
    &= E \Big\{ E\Big[ \{Y_i(t) - \mu_i(t; \boldsymbol{\beta}) \}w_i^{F}(t;\boldsymbol{\alpha}, \boldsymbol{\gamma}, \pi, h)dN_i(t) \big\rvert \boldsymbol{Z}_i(t), \boldsymbol{X}_i(t), \boldsymbol{W}_i, Y_i(t), C_i \ge t \Big] \big\rvert \boldsymbol{X}_i(t)\Big\} \\
    &= E \Big\{ E\Big[ \{Y_i(t) - \mu_i(t; \boldsymbol{\beta}) \}\left( \frac{D_i(t)}{\pi_i(\boldsymbol{W}_i; \alpha)} + \frac{(1 - D_i(t))}{(1 - \pi_i(\boldsymbol{W}_i; \alpha)} \right)\left(\frac{h(\boldsymbol{X}_i(t))}{\exp(\boldsymbol{Z}_i(t)\boldsymbol{\gamma})}\right)\\
    &\qquad \qquad \times dN_i(t)  \big\rvert \boldsymbol{Z}_i(t), \boldsymbol{X}_i(t), \boldsymbol{W}_i, Y_i(t), C_i \ge t \Big] \big\rvert \boldsymbol{X}_i(t)\Big\} \\
    &= E \Big\{  \{Y_i(t) - \mu_i(t; \boldsymbol{\beta}) \} E\Big[\left( \frac{D_i(t)}{\pi_i(\boldsymbol{W}_i; \alpha)} + \frac{(1 - D_i(t))}{(1 - \pi_i(\boldsymbol{W}_i; \alpha))} \right)\left(\frac{h(\boldsymbol{X}_i(t))}{\exp(\boldsymbol{Z}_i(t)\boldsymbol{\gamma})}\right)\\
    & \qquad \qquad \times dN_i(t)  \big\rvert \boldsymbol{Z}_i(t), \boldsymbol{X}_i(t), \boldsymbol{W}_i, Y_i(t), C_i \ge t \Big] \big\rvert \boldsymbol{X}_i(t)\Big\}.
\end{aligned}
\end{equation*}

If we assume the observation and treatment assignment processes are conditionally independent given the outcome, censoring time, and the covariates related to the outcome, probability treatment assignment, and observation intensity, then the conditional expectation can be separated into two multiplicative terms. That is, we can write the expectation above as

\begin{equation*}
\begin{aligned}
    = E \Big\{  \{Y_i(t) - &\mu_i(t; \boldsymbol{\beta}) \} E\Big[\left( \frac{D_i(t)}{\pi_i(\boldsymbol{W}_i; \alpha)} + \frac{(1 - D_i(t))}{(1 - \pi_i(\boldsymbol{W}_i; \alpha))} \right)\big\rvert \boldsymbol{Z}_i(t), \boldsymbol{X}_i(t), \boldsymbol{W}_i, Y_i(t), C_i \ge t \Big] \times \\
    &E\Big[\left(\frac{h(\boldsymbol{X}_i(t))}{\exp(\boldsymbol{Z}_i(t)\boldsymbol{\gamma})}\right)dN_i(t)  \big\rvert \boldsymbol{Z}_i(t), \boldsymbol{X}_i(t), \boldsymbol{W}_i, Y_i(t), C_i \ge t \Big] \big\rvert \boldsymbol{X}_i(t)\Big\}.
\end{aligned}
\end{equation*}

Recall that $D_i(t)$ is a covariate contained in the set $\boldsymbol{X}_i(t)$. As such, we can simplify this expression to be 

\begin{equation*}
\begin{aligned}
    =E \Big\{  \{Y_i(t) - &\mu_i(t; \boldsymbol{\beta}) \} \left( \frac{(D_i(t)}{\pi_i(\boldsymbol{W}_i; \alpha)} + \frac{(1 - D_i(t))}{(1 - \pi_i(\boldsymbol{W}_i; \alpha))} \right) \times  \\
    &\left(\frac{h(\boldsymbol{X}_i(t))}{\exp(\boldsymbol{Z}_i(t)\boldsymbol{\gamma})}\right)E\left[ dN_i(t) \big\rvert \boldsymbol{Z}_i(t), \boldsymbol{X}_i(t), \boldsymbol{W}_i, Y_i(t), C_i \ge t \right] \big\rvert \boldsymbol{X}_i(t)\Big\}.
\end{aligned}
\end{equation*}
As $C_i \ge t$, $dN_i(t)$ = $dN_i^*(t)$, and we can write this quantity as
\begin{equation*}
\begin{aligned}
    =E \Big\{  \{Y_i(t) - &\mu_i(t; \boldsymbol{\beta}) \} \left( \frac{(D_i(t)}{\pi_i(\boldsymbol{W}_i; \alpha)} + \frac{(1 - D_i(t))}{(1 - \pi_i(\boldsymbol{W}_i; \alpha))} \right) \times  \\
    &\left(\frac{h(\boldsymbol{X}_i(t))}{\exp(\boldsymbol{Z}_i(t)\boldsymbol{\gamma})}\right)E\left[ dN_i^*(t) \big\rvert \boldsymbol{Z}_i(t), \boldsymbol{X}_i(t), \boldsymbol{W}_i, Y_i(t), C_i \ge t \right] \big\rvert \boldsymbol{X}_i(t)\Big\}.
\end{aligned}
\end{equation*}
Then under \ref{O1}, we can simplify this expression to
\begin{equation*}
\begin{aligned}
    =E \Big\{  \{Y_i(t) - &\mu_i(t; \boldsymbol{\beta}) \} \left( \frac{(D_i(t)}{\pi_i(\boldsymbol{W}_i; \alpha)} + \frac{(1 - D_i(t))}{(1 - \pi_i(\boldsymbol{W}_i; \alpha))} \right) \times  \\
    &\left(\frac{h(\boldsymbol{X}_i(t))}{\exp(\boldsymbol{Z}_i(t)\boldsymbol{\gamma})}\right)E\left[ dN_i^*(t) \big\rvert \boldsymbol{Z}_i(t) \right] \big\rvert \boldsymbol{X}_i(t)\Big\}.
\end{aligned}
\end{equation*}
and under \ref{O4}, we can further simplify this to
\begin{equation*}
\begin{aligned}
    &=E \Big\{  \{Y_i(t) - \mu_i(t; \boldsymbol{\beta}) \} \left( \frac{D_i(t)}{\pi_i(\boldsymbol{W}_i; \alpha)} + \frac{(1 - D_i(t))}{(1 - \pi_i(\boldsymbol{W}_i; \alpha))} \right)h(\boldsymbol{X}_i(t))\lambda_0(t) \big\rvert \boldsymbol{X}_i(t)\Big\}\\
    &= \left( \frac{D_i(t)}{\pi_i(\boldsymbol{W}_i; \alpha)} + \frac{(1 - D_i(t))}{(1 - \pi_i(\boldsymbol{W}_i; \alpha))} \right)h(\boldsymbol{X}_i(t))\lambda_0(t)E \Big\{  \{Y_i(t) - \mu_i(t; \boldsymbol{\beta}) \}  \big\rvert \boldsymbol{X}_i(t)\Big\}.\\
\end{aligned}
\end{equation*}
 As $E \Big\{  \{Y_i(t) - \mu_i(t; \boldsymbol{\beta}) \}  \big\rvert \boldsymbol{X}_i(t)\Big\} = 0$, we have shown that the expectation of the estimating equation is 0 at $\{\boldsymbol{\beta}, \boldsymbol{\alpha}, \boldsymbol{\gamma}, \pi, h\}$ for any function $h(\cdot)$ of covariates $\boldsymbol{X}_i(t)$. We additionally note that $dN_i(t) = 0$ when the individual is censored ($C_i < t)$.
\end{proof}

\section{Additional Simulation Results}\label{sec:appendixsim}
\subsection{Additional Results for Simulation I}\label{sec:appendixsimI}

\begin{figure}[ht!]
    \centering
    \includegraphics[width = \textwidth]{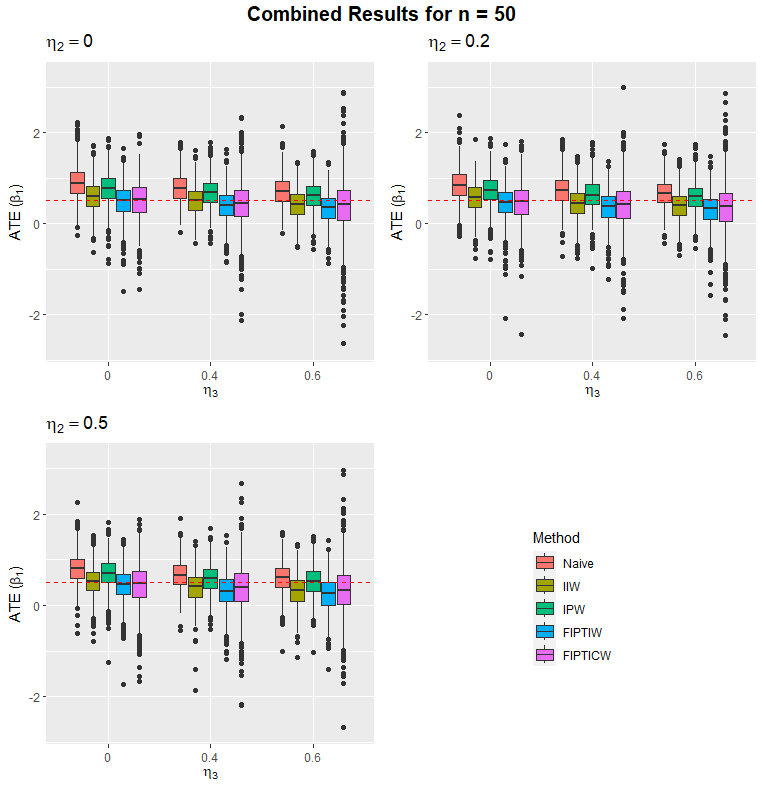}
    \caption{Results of Simulation I for $n$ = 50. Average treatment effect (ATE, $\beta_1$) is calculated by fitting an independent GEE with various weights for each simulation scheme over the 1000 generated data sets. The weighting methods include no weighting (unweighted), inverse intensity weighting (IIW), inverse probability of treatment weighting (IPTW), flexible inverse probability of treatment and intensity weighting (FIPTIW) and flexible inverse probability of treatment weighting with inverse probability of censoring weights included (FIPTICW). The true value of the ATE is 0.5, and is denoted by the red horizontal line.}
    \label{fig:censoring_n50}
\end{figure}

\begin{figure}[ht!]
    \centering
    \includegraphics[width = \textwidth]{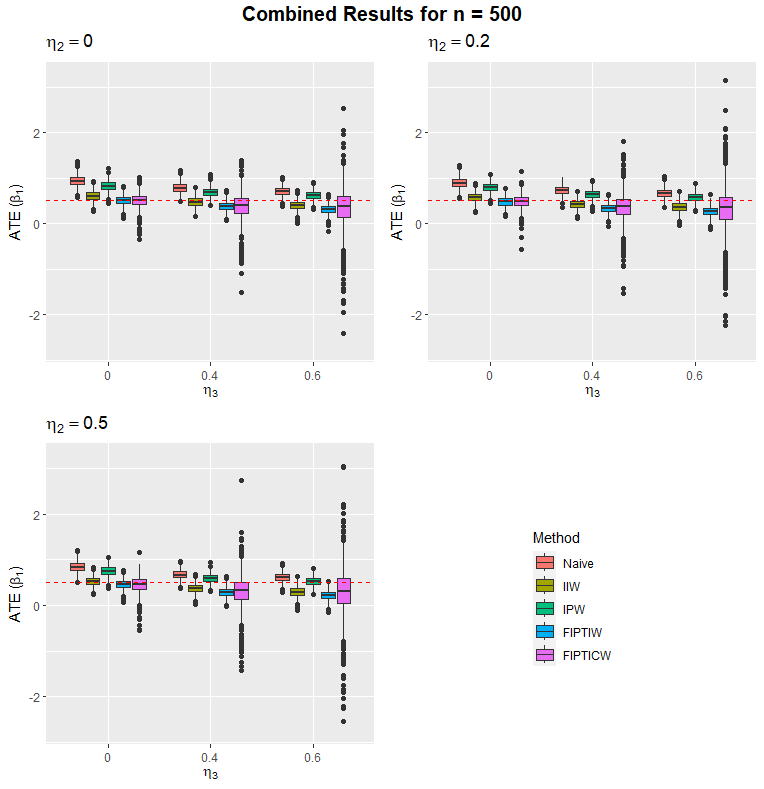}
    \caption{Results of Simulation I for $n$ = 500. Average treatment effect (ATE, $\beta_1$) is calculated by fitting an independent GEE with various weights for each simulation scheme over the 1000 generated data sets. The weighting methods include no weighting (unweighted), inverse intensity weighting (IIW), inverse probability of treatment weighting (IPTW), flexible inverse probability of treatment and intensity weighting (FIPTIW) and flexible inverse probability of treatment weighting with inverse probability of censoring weights included (FIPTICW). The true value of the ATE is 0.5, and is denoted by the red horizontal line.}
    \label{fig:censoring_n500}
\end{figure}

Simulation I was performed using sample sizes of $n = 50$ and $n = 500$. In both settings, when $\eta_1, \eta_2$ and $\eta_3$ are all zero, the noninformative censoring assumption is satisfied. In this setting, the FIPTIW and FIPTICW methods are unbiased while the other methods have large biases for both $n = 50$ and $n = 500$. The FIPTICW method, however, has a large variance. When $\eta_1$ is non-zero while $\eta_2$ and $\eta_3$ are zero, the FIPTIW and FIPTICW methods remain unbiased. 

As the magnitude of $\eta_3$ increases, more bias is introduced into our estimated ATEs under the FIPTIW and FIPTICW methods. In most cases, the bias and variance increase as the the magnitude of $\eta_3$ increases. However, the bias and MSE of the other methods do not tend to increase as $\eta_3$ increases. As $\eta_2$ increases, we see the bias and variance also tend to increase for the FIPTIW and FIPTICW methods. Again, the bias decreases as $\eta_2$ increases for the unweighted, IIW, and IPTW methods in many cases.

The FIPTIW and FIPTICW methods often had worse performance than the IIW and IPTW methods in terms of both bias and variance. Aside from when $\eta_2$ and $\eta_3$ are zero, the IIW and/or the IPTW methods outperform the FIPTIW and FIPTICW methods in most settings. We also note that often the best performing method in terms of bias and variance still results in biased estimation. 

In general, the FIPTICW method tends to have a smaller bias than the FIPTIW method. However, it also tends to have a higher variance and is still biased in many settings as the IIW weights are not adequately estimated. 

From this analysis, we see that there is again moderate sensitivity to the noninformative censoring assumption for the FIPTIW method, and further including the IPCW weights into the FIPTIW model into the model is not recommended. These results are similar to the $n = 100$ case, where variances are smaller for $n = 500$ and larger for $n = 50$.

\subsection{Additional Results for Simulation II}\label{sec:appendixsimII}

\begin{table}[ht]
    \centering
    \begin{tabular}{ll|lllllllll}
\multicolumn{2}{c}{} & \multicolumn{9}{c}{\textbf{Variables used to estimate intensity}} \\
$\boldsymbol{\gamma_2}$ & $\boldsymbol{\beta_2}$ &  & \textbf{Naive} & $\boldsymbol{D}$ & $\boldsymbol{G(t)}$ & $\boldsymbol{Z}$ & $\boldsymbol{D,G(t)}$ & $\boldsymbol{D,Z}$ & $\boldsymbol{G(t),Z}$ & $\boldsymbol{D,G(t),Z}$\\
\hline
0 & 0 &  &  &  &  &  &  &  &  & \\
 &  & Bias: & 0.277 & 0.277 & 0.275 & 0.079 & 0.274 & 0.024 & 0.079 & 0.024\\
 &  & MSE: & 0.118 & 0.118 & 0.117 & 0.04 & 0.117 & 0.034 & 0.04 & 0.034\\
0 & 2 &  &  &  &  &  &  &  &  & \\
 &  & Bias: & 0.281 & 0.281 & 0.284 & 0.081 & 0.281 & 0.026 & 0.081 & 0.027\\
 &  & MSE: & 0.142 & 0.142 & 0.151 & 0.063 & 0.144 & 0.057 & 0.062 & 0.058\\
0.3 & 0 &  &  &  &  &  &  &  &  & \\
 &  & Bias: & 0.298 & 0.298 & 0.297 & 0.093 & 0.297 & 0.035 & 0.094 & 0.036\\
 &  & MSE: & 0.133 & 0.133 & 0.131 & 0.046 & 0.131 & 0.038 & 0.046 & 0.038\\
0.3 & 2 &  &  &  &  &  &  &  &  & \\
 &  & Bias: & 0.361 & 0.361 & 0.297 & 0.154 & 0.291 & 0.097 & 0.085 & 0.03\\
 &  & MSE: & 0.199 & 0.199 & 0.159 & 0.08 & 0.149 & 0.063 & 0.059 & 0.053\\
\end{tabular}
    \caption{Simulation results for Simulation II for $n = 50$.  Bias and mean squared error (MSE) of the average treatment effect (ATE) is calculated by weighting the outcome model in Equation (\ref{eq:simulationmarginalmodel}) by inverse intensity weighting (IIW) for each simulation scheme over the 1000 generated data sets. Variables included in the IIW model are listed in the table. The true value of the ATE is 0.5.}
    \label{tab:simII_n50}
\end{table}

\begin{table}[ht]
    \centering
    \begin{tabular}{ll|lllllllll}
\multicolumn{2}{c}{} & \multicolumn{9}{c}{\textbf{Variables used to estimate intensity}} \\
$\boldsymbol{\gamma_2}$ & $\boldsymbol{\beta_2}$ &  & \textbf{Naive} & $\boldsymbol{D}$ & $\boldsymbol{G(t)}$ & $\boldsymbol{Z}$ & $\boldsymbol{D,G(t)}$ & $\boldsymbol{D,Z}$ & $\boldsymbol{G(t),Z}$ & $\boldsymbol{D,G(t),Z}$\\
\hline
0 & 0 &  &  &  &  &  &  &  &  & \\
 &  & Bias: & 0.297 & 0.297 & 0.297 & 0.084 & 0.297 & 0.027 & 0.084 & 0.027\\
 &  & MSE: & 0.093 & 0.093 & 0.093 & 0.011 & 0.093 & 0.004 & 0.011 & 0.004\\
0 & 2 &  &  &  &  &  &  &  &  & \\
 &  & Bias: & 0.298 & 0.298 & 0.298 & 0.085 & 0.298 & 0.028 & 0.085 & 0.029\\
 &  & MSE: & 0.095 & 0.095 & 0.096 & 0.012 & 0.095 & 0.006 & 0.012 & 0.006\\
0.3 & 0 &  &  &  &  &  &  &  &  & \\
 &  & Bias: & 0.299 & 0.299 & 0.299 & 0.081 & 0.299 & 0.023 & 0.081 & 0.022\\
 &  & MSE: & 0.094 & 0.094 & 0.093 & 0.01 & 0.093 & 0.004 & 0.01 & 0.004\\
0.3 & 2 &  &  &  &  &  &  &  &  & \\
 &  & Bias: & 0.38 & 0.38 & 0.315 & 0.162 & 0.314 & 0.103 & 0.092 & 0.033\\
 &  & MSE: & 0.151 & 0.151 & 0.106 & 0.032 & 0.104 & 0.016 & 0.013 & 0.006\\
\end{tabular}
    \caption{Simulation results for Simulation II for $n = 500$.  Bias and mean squared error (MSE) of the average treatment effect (ATE) is calculated by weighting the outcome model in Equation (\ref{eq:simulationmarginalmodel}) by inverse intensity weighting (IIW) for each simulation scheme over the 1000 generated data sets. Variables included in the IIW model are listed in the table. The true value of the ATE is 0.5.}
    \label{tab:simII_n500}
\end{table}

When $\gamma_2 = 0$ and $\beta_2 = 0$, the covariate $G(t)$ is not related to the observation intensity or longitudinal outcome and any weighting model that does not include the confounder $Z$ and treatment $D$ are biased. Further including an unrelated covariate $G(t)$ does not impact the MSE in this setting.

When $\gamma_2 = 0$ and $\beta_2 = 2$, the covariate $G(t)$ is related only to the longitudinal outcome. Again, any weighting model that does not include the confounder $Z$ and treatment $D$ are biased. Further including the covariate only related to the longitudinal outcome, $G(t)$, did not influence the bias or MSE of the ATE. Similar results are seen when $\gamma_2 = 0.3$ and  $\beta_2 = 0$ (the covariate $G(t)$ is only related to the observation intensity). 

When both $\gamma_2$ and $\beta_2$ are non-zero, $G(t)$ is related to both the observation and outcome processes but is unrelated to the treatment assignment. In this setting, the estimate is only unbiased when the weighting model includes all three covariates.

This simulation shows the importance of including observation process confounders in the intensity model. Including covariates that do not confound the relationship between the observation times and the outcome did not impact estimation greatly, with no substantial increases in MSE.

\subsection{Additional Results for Simulation III}\label{sec:appendixsimIII}

When the underlying processes both have low degrees of informativeness, we see low mean proportions of average weights and the relative bias monotonically decreases as the percentile at which we trim extreme weights increases. That is, we see the smallest relative bias when no trimming is performed, which corresponds to trimming at the 100th percentile (RB = 0.042). However, the variance is minimized at the 91st percentile when weights are trimmed individually prior to multiplication. In terms of MSE, trimming weights prior to multiplication at the 98th percentile results in the smallest MSE (MSE = 0.025). When no trimming is performed, the MSE is only slightly larger at 0.026.

When the treatment assignment process has moderate informativeness and the informativeness in the observation process remains low, the relative bias is no longer minimized when no trimming is performed. The relative bias is minimized when trimming the FIPTIW weights after multiplication at the 95th percentile (RB = 0.000) and the MSE is minimized when trimming the FIPTIW weights after multiplication between the 91st to 95th percentiles (MSE = 0.017). Under no trimming, the relative bias is -0.056 and the MSE is 0.022. When the treatment assignment has high informativeness where the relative bias is minimized when trimming the FIPTIW weights before multiplication at the 96th percentile (RB = 0.000) and the MSE is minimized when trimming the FIPTIW weights after multiplication at the 94th or 95th percentiles (MSE = 0.018). In this simulation, without trimming the relative bias is -0.060 and the MSE is 0.022. These results show that weight trimming can be used to reduce the overall bias of our estimated ATE, but there is not a large difference in the overall MSE when trimming at the optimal threshold versus not trimming. 

When the observation process has moderate or high degrees of informativeness, we see more extremity in both the IIW and IPTW (and thus FIPTIW) weights. The shapes of the relative bias curves differ greatly from the first three bias plots in Figure \ref{fig:trimmingresults}. We see the relative bias, variance, and MSE curves are more similar between trimming before and after methods. When the treatment assignment process has low informativeness and the observation process has moderate informativeness, the relative bias is minimized by either trimming the FIPTIW weights before multiplying between the 93rd and 100th percentiles, or trimming the FIPTIW weights after multiplying at the 100th percentile (RB = -0.016). In this setting, the variance tends to decrease as the threshold percentile increases. The minimum MSE is obtained when trimming the FIPTIW weights before multiplying between the 97th and 100th percentiles, or trimming the FIPTIW weights after multiplying between the 88th and 100th percentiles (MSE = 0.024). We see similar results when the observation process is highly informative where the minimum relative bias is achieved when trimming the FIPTIW weights prior to multiplying at \ 100th percentile, or trimming the FIPTIW weights after multiplying between the 96th and 100th percentiles (RB = -0.015). The MSE is minimized when trimming the FIPTIW weights before multiplying between the 98th and 100th percentiles or trimming the FIPTIW weights after multiplying between the 93rd and 100th percentiles (MSE = 0.024). In this setting, performing no weight trimming is one of the optimal strategies in terms of both relative bias and MSE. 

When both processes are moderately informative, the bias is minimized when trimming weights before multiplication at the 94th percentile, or after multiplication at the 92nd percentile (RB = 0.000). Without weight trimming, the relative bias is 0.049. Again, we see the variance tend to decrease as the trimming threshold increases. As such, the MSE is minimized when trimming the FIPTIW weights before multiplication between the 93rd to 96th percentiles, or when trimming the FIPTIW weights after multiplication at the 93rd percentile (MSE = 0.018). Without weight trimming the MSE is 0.021. Again, these results show that weight trimming may reduce the overall bias of our estimated ATE, but there is not a large difference in the overall MSE when trimming at the optimal threshold compared to no trimming.

\section{Sensitivity Analyses for the Malaria Data Set}\label{sec:appendixSA}

\subsection{Ignoring Censoring}\label{sec:appendixSAcensoring}
The analysis in Section \ref{sec:FIPTIWdatanalysis} was repeated without the use of artificial censoring. The same 287 individuals were included in the study, but we considered all observations in the data set. In this data set, 69.34\% of individuals were censored for nonignorable reasons (i.e., reasons potentially related to the longitudinal outcome). The number of observations per individual ranged from 1 to 124, with a mean number of observations of 33.19. 

The variables included in the propensity score model were the same as in Section \ref{sec:FIPTIWdatanalysis}. The maximum IPTW weight was 11.35, and 4.67\% of IPTW weights were above 5, and 0.21\% above 10. The same variables were included in the intensity model as in Section \ref{sec:FIPTIWdatanalysis}, and the maximum IIW weight was 11.16. 2.78\% of weights were above 5, 0.71\% of weights were above 10, and 0.01\% were above 20. After multiplication, the maximum FIPTIW weight was 36.22.  24.30\% of FIPTIW weights were above 5, 6.65\% were above 10, and 1.37\% were above 20.  When employing weight trimming, any FIPTIW weight above 12.03 was trimmed.

\begin{table}[]
    \centering
    \begin{tabular}{llll|ll}
\textbf{Weighting Method} & $\eta_1$ & \textbf{SE}($\eta_1$) & \textbf{95\% CI for} $\eta_1$ & \textbf{Odds Ratio (OR) }& \textbf{95\% CI for OR}\\
  \hline
None & 0.730 & 0.061 & (0.611, 0.848) & 2.075 & (1.843, 2.336)\\
IPTW & 0.613 & 0.058 & (0.5, 0.726) & 1.846 & (1.649, 2.067)\\
IIW & 0.911 & 0.061 & (0.791, 1.032) & 2.487 & (2.206, 2.806)\\
FIPTIW & 0.773 & 0.059 & (0.656, 0.889) & 2.166 & (1.928, 2.433)\\
FIPTIW (Trimmed) & 0.784 & 0.059 & (0.669, 0.900) & 2.190 & (1.952, 2.458)\\
\\
\end{tabular}
    \caption{Results of the estimation of the odds ratio (OR) of malaria diagnoses for those residing in households with unprotected water sources versus those residing in households with protected water sources ($\beta_1$) by different weighting methods for an independent GEE in the PRISM cohort. Sample includes those who were censored, and likely violates the noninformative censoring violation.}
    \label{tab:results_malaria_SA_censoring}
\end{table}

The results are given in Table \ref{tab:results_malaria_SA_censoring}. The untrimmed FIPTIW model estimates an 2.17 times increase in the odds of malaria diagnosis for those residing in households with unprotected versus protected water sources. The result is similar, but slightly larger, when weight trimming is employed. When comparing to the analysis in Section \ref{sec:FIPTIWdatanalysis}, the estimates are much larger. We again see sensitivity to violations of the noninformative censoring assumption, as in Simulation II in Section \ref{sec:appendixsimII}. As such, analysts must be cautious when individuals who were censored for reasons that may be related to the longitudinal outcome are included in the analysis.

\subsection{Ignoring Clustering}\label{sec:appendixSAclustering}

The analysis performed in Section \ref{sec:FIPTIWdatanalysis} was repeated with the inclusion of all individuals in each family. This resulted in a sample size of 725 children in 287 unique households. Only 1.5\% of individuals were censored for reasons potentially related to the longitudinal outcome when artificial censoring at 6 months was used. The baseline demographics for this sample are given in Table \ref{tab:malariaSAclustering}. 

The variables included in the propensity score model were the same as those described in Section \ref{sec:FIPTIWdatanalysis}. The maximum IPTW weight was 15.11, where 6.63\% of IPTW weights were above 5 and 0.17\% were above 10.. The same variables were included in the intensity model as in Section \ref{sec:FIPTIWdatanalysis}, and the maximum IIW weight was 4.15. After multiplying the weights together, the maximum FIPTIW weight was 18.14. 8.95\% of FIPTIW weights were above 5, and 2.01\% were above 10. When employing weight trimming, any FIPTIW weight above 6.73 was trimmed.

\begin{table}[ht]
    \centering
    \begin{tabular}{ll|ll}
\textbf{Covariate} & \textbf{Mean (SD) or n (\%)} & \textbf{Covariate} & \textbf{Mean (SD) or n (\%)}\\
\hline
&&&\\
\textbf{Age at Enrollment} & 5.56 (2.13)&\multicolumn{2}{l}{\textbf{Unprotected Water Source}}\\
\multicolumn{2}{l|}{\textbf{Sex}}&\hspace{1em}No & 542 (74.76\%)\\
\hspace{1em}Female & 358 (49.38\%)&\hspace{1em}Yes & 183 (25.24\%)\\
\hspace{1em}Male & 367 (50.62\%)&\multicolumn{2}{l}{\textbf{Dwelling Type}}\\
\multicolumn{2}{l|}{\textbf{Sub-County}}&\hspace{1em}Modern & 206 (28.41\%)\\
\hspace{1em}Walukuba & 198 (27.31\%)&\hspace{1em}Traditional & 519 (71.59\%)\\
\hspace{1em}Kihihi & 272 (37.52\%)&\multicolumn{2}{l}{\textbf{Food Problems per Week}}\\
\hspace{1em}Nagongera & 255 (35.17\%)&\hspace{1em}Sometimes & 262 (36.14\%)\\
\multicolumn{2}{l}{\textbf{Household Wealth Index}}&\hspace{1em}Never & 115 (15.86\%)\\
\hspace{1em}Least poor & 235 (32.41\%)&\hspace{1em}Often & 100 (13.79\%)\\
\hspace{1em}Poorest & 257 (35.45\%)&\hspace{1em}Seldom & 114 (15.72\%)\\
\hspace{1em}Middle & 233 (32.14\%)&\hspace{1em}Always & 134 (18.48\%)\\
\multicolumn{2}{l|}{\textbf{Drinking Water Source}}&\multicolumn{2}{l}{\textbf{Waste Facilities}}\\
\hspace{1em}Public tap & 249 (34.34\%)&\hspace{1em}{}Covered pit latrine, no slab & 244 (33.66\%)\\
\hspace{1em}Protected public well & 63 (8.69\%)&\hspace{1em}{}Covered pit latrine with slab & 44 (6.07\%)\\
\hspace{1em}River/stream & 68 (9.38\%)&\hspace{1em}{}Composting  toilet & 9 (1.24\%)\\
\hspace{1em}Protected spring & 61 (8.41\%)&\hspace{1em}{}Uncovered pit latrine, no slab & 298 (41.1\%)\\
\hspace{1em}Borehole & 169 (23.31\%)&\hspace{1em}{}Flush toilet & 21 (2.9\%)\\
\hspace{1em}Open public well & 61 (8.41\%)&\hspace{1em}{}Uncovered pit latrine with slab & 15 (2.07\%)\\
\hspace{1em}Pond/lake & 15 (2.07\%)&\hspace{1em}{}Vip latrine & 8 (1.1\%)\\
\hspace{1em}Unprotected spring & 39 (5.38\%)&\hspace{1em}{}No facility & 86 (11.86\%)\\
\textbf{Number of Persons Living in House} & 6.75 (2.93)&\\
\\
\end{tabular}
    \caption{Baseline demographics for the 725 children when we allow multiple children from the same household to be included in the analysis.}
    \label{tab:malariaSAclustering}
\end{table}

\begin{table}[ht]
    \centering
    \begin{tabular}{l|rrrrr}
  \textbf{Weighting Method} & $\boldsymbol{\eta_1}$ & \textbf{SE(}$\boldsymbol{\eta_1}$\textbf{)} & \textbf{95\% CI for} $\boldsymbol{\eta_1}$ & \textbf{Odds Ratio (OR)} & \textbf{95\% CI for OR} \\
  \hline
None & 0.456 & 0.110 & (0.240, 0.673) & 1.578 & (1.271, 1.959)\\
IPTW & 0.335 & 0.102 & (0.135, 0.535) & 1.398 & (1.145, 1.707)\\
IIW & 0.498 & 0.107 & (0.288, 0.709) & 1.645 & (1.333, 2.032)\\
FIPTIW & 0.375 & 0.100 & (0.180, 0.570) & 1.455 & (1.197, 1.768)\\
FIPTIW (Trimmed) & 0.360 & 0.100 & (0.164, 0.556) & 1.433 & (1.178, 1.743)\\
\end{tabular}
    \caption{Results of the estimation of the odds ratio (OR) of malaria diagnoses for those residing in households with unprotected water sources versus those residing in households with protected water sources ($\beta_1$) by different weighting methods for an independent GEE in the PRISM cohort. Sample includes those who resided in the same household, and likely violates the independence assumption.}
    \label{tab:results_malaria_SA_clustering}
\end{table}

The results are given in Table \ref{tab:results_malaria_SA_clustering}. Based on the distribution of the weights, we see some extremity and recommend weight trimming in this setting. The untrimmed FIPTIW model estimates the odds of malaria diagnosis for those residing in households with unprotected versus protected water sources is 1.45 times greater. The trimmed FIPTIW OR estimate is similar in magnitude (OR = 1.43). The estimated effect sizes are slightly smaller than in the analysis presented in Section \ref{sec:FIPTIWdatanalysis}. There appears to be some sensitivity to violations of the independence assumption and as such, we identify extending existing methodology to account for the correlations between individuals as an area of future work.

\subsection{Employing Random Censoring for Individuals Censored Prior to 6 Months}\label{sec:appendixSArandom}

The analysis performed in Section \ref{sec:FIPTIWdatanalysis} was repeated where individuals who were censored prior to 6 months were randomly censored between 0 and 6 months, instead of at 6 months exactly. 

The same 287 children were included as in the original analysis. However, the distribution of observation times varied slightly due to the introduction of random sampling for a small sample of individuals who were likely censored prior to 6 months. The number of observations per individual ranged from 1 to 11, with an average of 4.78 clinic visits. Of the 1,373 observations, 287 were for enrollment, 516 were scheduled, and 570 were unscheduled. Malaria was diagnosed in 13.03\% of observations. 38.68\% of patients had a Malaria diagnosis at some point. 

The same covariates in the propensity score (IPW) and observation intensity (IIW) models were used in this analysis as described in Section \ref{sec:FIPTIWdatanalysis}. Table \ref{tab:results_malaria_SA_randomcensoring} show the results of the analysis. 

\begin{table}[ht]
    \centering
    \begin{tabular}{l|rrrrr}
  \textbf{Weighting Method} & $\boldsymbol{\eta_1}$ & \textbf{SE(}$\boldsymbol{\eta_1}$\textbf{)} & \textbf{95\% CI for} $\boldsymbol{\eta_1}$ & \textbf{Odds Ratio (OR)} & \textbf{95\% CI for OR} \\
  \hline
None & 0.423 & 0.186 & (0.059, 0.786) & 1.527 & (1.060, 2.196)\\
IPTW & 0.323 & 0.168 & (-0.007, 0.652) & 1.381 & (0.993, 1.919)\\
IIW & 0.558 & 0.178 & (0.210, 0.906) & 1.747 & (1.233, 2.475)\\
FIPTIW & 0.428 & 0.168 & (0.099, 0.757) & 1.534 & (1.104, 2.132)\\
FIPTIW (Trimmed) & 0.410 & 0.166 & (0.084, 0.736) & 1.507 & (1.088, 2.087)\\
\\
\end{tabular}
    \caption{Results of the estimation of the odds ratio (OR) of malaria diagnoses for those residing in households with unprotected water sources versus those residing in households with protected water sources ($\beta_1$) by different weighting methods for an independent GEE in the PRISM cohort. In this analysis, those who were assumed to have been censored prior to 6 months were randomly censored between 0 and 182.5 days.}
    \label{tab:results_malaria_SA_randomcensoring}
\end{table}

The results of this analysis show are similar to those shown in Table \ref{tab:malariaresults1}, showing very little sensitivity to including a small proportion of individuals who dropped out of the study (potentially for reasons related to the longitudinal outcome) before the study end date.

\end{document}